\documentclass[sn-mathphys,Numbered]{sn-jnl}
\usepackage{graphicx}%
\usepackage{multirow}%
\usepackage{amsmath,amssymb,amsfonts}%
\usepackage{amsthm}%
\usepackage{mathrsfs}%
\usepackage[title]{appendix}%
\usepackage{xcolor}%
\usepackage{textcomp}%
\usepackage{manyfoot}%
\usepackage{booktabs}%
\usepackage{algorithm}%
\usepackage{algorithmicx}%
\usepackage{algpseudocode}%
\usepackage{listings}%
\usepackage[all]{xy}

\newtheorem{theorem}{Theorem}
\newtheorem{proposition}[theorem]{Proposition}%
\newtheorem{lemma}[theorem]{Lemma}%
\newtheorem{cor}[theorem]{Corollary}%

\newtheorem{definition}{Definition}%

\providecommand{\norm}[1]{\lVert#1\rVert}
\providecommand{\abs}[1]{\lvert#1\rvert}
\providecommand{\inner}[2]{\langle#1,#2\rangle}

\title{Quantization as a Categorical Equivalence for Hilbert Bimodules and Lagrangian Relations}

\author{\fnm{Benjamin H.} \sur{~Feintzeig}}\email{bfeintze@uw.edu}

\affil{\centering{\orgdiv{Department of Philosophy}, \orgname{University of Washington},\\ \orgaddress{\city{Seattle}, \state{WA}, \country{USA}, \postcode{98195-3350}}}}

\abstract{
    It is well known that classical and quantum theories carry distinct types of \emph{representations}, each type of representation corresponding to possible values of generalized charges in the classical or quantum context.  This paper demonstrates a sense in the structure of these representation theories is preserved from classical to quantum physics.  To show this, I discuss distinct representation-theory preserving morphisms in the classical and quantum contexts.  Specifically, I consider categories whose morphisms are Lagrangian relations in the classical context and Hilbert bimodules in the quantum context.  These morphisms are significant because they give rise to \emph{induced representations} of classical and quantum theories, respectively.  I consider quantization and the classical limit as determining functors between these categories.  I treat quantization via the strict deformation quantization of a Poisson algebra and the classical limit via the extension of a uniformly continuous bundle of C*-algebras.  With these tools, I prove that the quantization and classical limit functors are ``almost-inverse" to each other, thus establishing a categorical equivalence.
}

\begin{document}

\maketitle

\section{Introduction}

The goal of this paper is to show a sense in which the structure for representing symmetries and their associated charges is shared in common between classical and quantum physics.  This structure is typically associated with the representation theory of a Lie group, where the appropriate notion of a representation is defined relative to the framework of classical or quantum physics (see \citep{Ma68,MaWe74}).  For my purposes, classical theories in a Hamiltonian framework are formulated on Poisson manifolds that carry representations as symplectic manifolds.  On the other hand, quantum theories are formulated with C*-algebras of observables that carry representations on Hilbert spaces.  Even if the classical and quantum theories have the same Lie group of symmetries, the classical and quantum representations may fail to stand in a one-to-one correspondence.  For example, charge in classical theories typically takes continuous values, leading to a continuous family of representations, while charge in quantum theories is typically ``quantized", leading to a discrete family of representations.  This leads to the question of whether one can make precise the sense in which, nevertheless, the classical and quantum representation theories still stand in correspondence with each other.

I hope to show that one can indeed make this structural correspondence precise in the language of category theory for a certain class of models of classical physics and their quantized counterparts.  I will consider a category of models of classical physics in which the objects are Poisson manifolds associated with Lie groupoids and the arrows are Lagrangian relations, understood as morphisms that preserve classical representation-theoretic structure.  I will restrict to a special case of Lagrangian relations determined by symplectic dual pairs.  Likewise, I will consider a category of models of quantum physics in which the objects are C*-algebras associated as the quantization of a Lie groupoid and the arrows are (equivalence classes of) Hilbert bimodules, understood as morphisms that preserve quantum representation-theoretic structure.  Previously, \citet{La02a} has defined a quantization functor $Q$ from the classical category to the quantum category by associating classical morphisms (symplectic dual pairs) to quantum morphisms (Hilbert bimodules).  On the other hand, \citet{FeSt24} have defined a classical limit functor $L$ from the quantum category to the classical category by associating quantum morphisms to classical morphisms.  In this paper, I will show that the classical limit functor is almost inverse to the quantization functor, in the sense that $L\circ Q$ and $Q\circ L$ are each naturally isomorphic to the identity functor.  It is this result, generalizing previous work \citep{Fe24}, that I take to substantiate the claim that classical and quantum theories share the representation-theoretic structures associated with physical charges and symmetries.

The paper proceeds as follows.  In \S\ref{sec:objects}, I define the objects of the classical and quantum categories for models constructed from Lie groupoids \citep{Ma87,La98b}.  A classical object is the dual of a Lie algebroid associated with a Lie groupoid, with its natural Poisson structure \citep{La01a,We00,CrFe04}. A quantum object is a (reduced) Lie groupoid C*-algebra \citep{La98b,La01a} associated with a Lie groupoid.  I will also define the actions of the quantization and classical limit functors on these objects.  Quantization is achieved via strict deformation quantization \citep{Ri89,Ri93,La98b,La99}, while the classical $\hbar\to 0$ limit is understood via the extension of a uniformly continuous bundle of C*-algebras \citep{KiWa95,StFe21a,StFe21}.  In \S\ref{sec:arrows}, I define the classical and quantum arrows \citep{La01b,La01c}.  The classical arrows are Lagrangian relations determined by symplectic dual pairs \citep{We83,Xu91,Xu92,BuWe04,We10} while the quantum arrows are (equivalence classes of) Hilbert bimodules \citep{La95,RaWi98}.  I will review the transformations involving the quantization and classical limit of these arrows provided by \citet{La02a} and \citet{FeSt24}, respectively.  In \S\ref{sec:equivalence}, I prove the central result of this paper: the quantization and classical limit functors form a categorical equivalence.  Finally, \S\ref{sec:con} provides a discussion of the main result, as well as a summary and conclusion.

\section{Classical and Quantum Models}
\label{sec:objects}

\subsection{Systems associated to Lie groupoids}

The classical and quantum models that we consider encompass a range of physical systems that can be represented by Lie groupoids \citep{Ma87,We00}.  A Lie groupoid is a small category $G$ in which every arrow has an inverse.  We identify the objects in $G$ with the identity arrows, the set of which we denote by $G_0$.  There is a source projection $s_G: G\to G_0$ and a target projection $t_G: G\to G_0$.  Groupoid multiplication is defined as composition of arrows $x,x'\in G$, denoted by $x\cdot x'$ when $s_G(x) = t_G(x')$.  Moreover, for a Lie groupoid, both $G$ and $G_0$ are smooth manifolds, $s_G$ and $t_G$ are surjective submersions, multiplication is a smooth map $G\times G\to G$, and the inclusion $G_0\hookrightarrow G$ is smooth.  As discussed in \citet{La98b,La99}, examples of Lie groupoids include the pair groupoid of any manifold, the action groupoid for a Lie group acting on a manifold, and the gauge groupoid for a principal bundle.

A Lie groupoid can be associated with a Poisson manifold understood as the phase space of a classical mechanical system.  The Poisson manifold at issue is the dual of the Lie algebroid.  We denote the Lie algebroid for $G$ by $\mathfrak{G}\rightrightarrows^{TG_0}_{G_0}$.  With $T^tG = \ker((t_G)_*)$ a subbundle of $TG$, we have $\mathfrak{G}\cong T^tG_{|{G_0}}$, understood as the restriction to $G_0$.  In other words, $(x,\xi)\in T^tG_{|G_0}$ means that $x\in G_0$ is an identity arrow in $G$ and $\xi\in T_xG$ is a tangent vector with $(t_G)_*(\xi)=0$.  The Lie algebroid $\mathfrak{G}$ is a vector bundle over $G_0$, with the projection $\tau_G:\mathfrak{G}\to G_0$ given by $\tau_G(x,\xi) = s_G(x) = t_G(x)$ for all $(x,\xi)\in T^tG_{|G_0}$.  The anchor map $\tau^a_G: \mathfrak{G}\to TG_0$ is given by the restriction of $(s_G)_*$.  The dual Lie algebroid $\mathfrak{G}^*$ is the dual vector bundle to $\mathfrak{G}$ over $G_0$ with projection denoted by $\tau^*_G$.

The manifold $\mathfrak{G}^*$ carries a Poisson bracket on $C^\infty(\mathfrak{G}^*)$ \citep{La98b,La99}.  The Poisson bracket is uniquely determined by relations among special cases of functions.  First, consider functions that only vary along the base space and not the fibers: take $\tilde{f},\tilde{g}\in C^\infty(G_0)$ and regard them as functions $f,g\in C^\infty(\mathfrak{G}^*)$ by defining $f(\eta) = \tilde{f}(\tau_G^*(\eta))$ and likewise $g(\eta) = \tilde{g}(\tau_G^*(\eta))$ for $\eta\in\mathfrak{G}^*$.  The Poisson bracket of $f$ and $g$ is defined as
\begin{align}
    \{f,g\} &= 0.
    \end{align}
Next, consider functions that are dual to sections of $\mathfrak{G}$: take sections $s_1,s_2: G_0\to \mathfrak{G}$ and define $\tilde{s}_1,\tilde{s}_2\in C^\infty(\mathfrak{G}^*)$ by $\tilde{s}_1(\eta) = \eta(s_1(\tau^*_G(\eta)))$ and $\tilde{s}_2(\eta) = \eta(s_2(\tau^*_G(\eta)))$ for $\eta\in \mathfrak{G}^*$.  The Poisson bracket for $\tilde{s}_1$ with $f$ as above and $\tilde{s}_2$ are defined as
\begin{align}
    \{\tilde{s}_1,f\} &= -(\tau^a_G\circ s_1)(\tilde{f})\\
    \{\tilde{s}_1,\tilde{s}_2\} &= -\widetilde{[s_1,s_2]}_{\mathfrak{G}},
\end{align}
where $[\cdot,\cdot]_{\mathfrak{G}}$ denotes the Lie bracket.  As mentioned, these relations completely define the Poisson bracket on $\mathfrak{G}^*$.  This provides enough structure to consider $\mathfrak{G}^*$ as the phase space of a Hamiltonian system.

A Lie groupoid also defines a non-commutative C*-algebra of observables of a quantum system \citep{Re80,La98b,La01a}.  The reduced\footnote{Often the reduced groupoid C*-algebra is denoted $C_r^*(G)$ to distinguish it from the full groupoid C*-algebra.  The two coincide when the groupoid is amenable.  In this paper, we do not consider the full groupoid C*-algebra so we continue to use the notation $C^*(G)$ for the reduced groupoid C*-algebra for simplicity.} groupoid C*-algebra $C^*(G)$ is defined as a suitable completion of the algebra $C_c^\infty(G)$, with the convolution product.  To define these algebras, recall that the Lie groupoid $G$ carries associated right/left Haar systems $(\nu_{x_0}^{s/t})_{x_0\in G_0}$ defined on $s_G^{-1}(x_0)$ and $t_G^{-1}(x_0)$, respectively.  Consider also a given quasi-invariant measure $\nu_0$ on $G_0$.  Define the left regular representation $\pi^G$ of $C_c^\infty(G)$ on $L^2(G,\nu_0\times \nu^s)$ by
\begin{align}
\pi^G(f)(\psi)(x) = \int_{s_G^{-1}(s_G(x))} d\nu^s_{s_G(x)}(x')f(xx'^{-1})\psi(x')
\end{align}
for $f\in C_c^\infty(G)$ and $\psi\in L^2(G,\nu_0\times \nu^s)$.  Note that $\pi^G(f)$ is a Hilbert-Schmidt operator with integral kernel
\begin{align}
\label{eq:HSkernel}
    K^f(x,x') = f(xx'^{-1})
\end{align}
when $s_G(x) = s_G(x')$ and $K^f(x,x') = 0$ otherwise.  The reduced groupoid C*-algebra $C^*(G)$ is the completion of $C_c^\infty(G)$ in the operator norm for the representation $\pi^G$.

\subsection{Quantization and the Classical Limit for Lie Groupoids}

The classical model---the Poisson manifold $\mathfrak{G}^*$---is related to the quantum model---the C*-algebra $C^*(G)$---by \emph{strict deformation quantization}.  We begin by reviewing the definitions of the structures involved \citep{Ri89,Ri93,La98b,La06}, focusing on the notion of a (uniformly) continuous bundle of C*-algebras to formulate the classical limit \citep{KiWa95,StFe21,StFe21a}.\bigskip

\begin{definition}
    A \emph{strict quantization} of a Poisson algebra $\mathcal{P}\subseteq C_b(M)$ of bounded functions on a smooth Poisson manifold $M$ consists in:
    \begin{itemize}
        \item a family of C*-algebras $(\mathfrak{A}_\hbar)_{\hbar\in [0,1]}$ called the \emph{fibers}; and
        \item a family of \emph{quantization maps} $(\mathcal{Q}_\hbar: \mathcal{P}\to \mathfrak{A}_\hbar)_{\hbar\in [0,1]}$, with $\mathcal{Q}_0: \mathcal{P}\hookrightarrow C_b(M)$ the inclusion map
    \end{itemize}
    satisfying for all $f,g\in \mathcal{P}$:
    \begin{enumerate}
        \item (von Neumann's condition) $\lim_{\hbar\to 0}\norm{\mathcal{Q}_\hbar(f)\mathcal{Q}_\hbar(g) - \mathcal{Q}_\hbar(fg)} = 0$;
        \item (Dirac's condition) $\lim_{\hbar\to 0}\norm{\mathcal{Q}_\hbar\big(\{f,g\}\big) - [\mathcal{Q}_\hbar(f),\mathcal{Q}_\hbar(g)]} = 0$; and
        \item (Rieffel's condition) the map $\hbar\mapsto \norm{\mathcal{Q}_\hbar(f)}$ is continuous.
    \end{enumerate}
If furthermore, for each $\hbar\in [0,1]$, the map $\mathcal{Q}_\hbar$ is one-to-one and $\mathcal{Q}_\hbar(\mathcal{P})$ is dense in $\mathfrak{A}_\hbar$, then the structure is called a strict \emph{deformation} quantization.\footnote{The definition of a strict deformation quantization may be generalized so that the base space is not restricted to $I=[0,1]$, but includes also more general subsets $I\subseteq \mathbb{R}$ for which $0$ is an accumulation point.  In this paper, however, we will only employ the base spaces $[0,1]$ and $(0,1]$.}
\end{definition}\bigskip

\noindent We will sometimes refer to a strict deformation quantization by $\mathcal{Q}_\hbar$ for shorthand.

The primary example of a strict deformation quantization for the purposes of this paper is one that relates the model of a classical system as the dual of a Lie algebroid $\mathfrak{G}^*$ to the corresponding model of a quantum system as the groupoid C*-algebra $C^*(G)$. 
 We now summarize the details provided in \citet{La98b,La99} (See also \citet{LaRa00}).  To define the quantization maps, we employ two additional aspects of the Lie algebroid structure: the Weyl exponential map and the fiberwise Fourier transform.

To define the Weyl exponential, we suppose that the Lie algebroid $\mathfrak{G}$ is equipped with a connection.  Consider the bundle $\tau_{T^tG\to G}: T^tG\to G$ and define the left exponential map $\exp^L: \mathfrak{G}\to G$ by
\begin{align}
    \exp^L(X) = \tau_{T^tG\to G}(X'(1)),
\end{align}
where $X'(t)$ is the geodesic flow on $T^tG$, and the left exponential is defined when the geodesic flow is defined at $t=1$.  Then define the Weyl exponential map $\exp^W: \mathfrak{G}\to G$ by
\begin{align}
    \exp^W(X) = \exp^L\big(-X/2\big)^{-1}\exp^L\big(X/2\big).
\end{align}

To define the fiberwise Fourier transform, we consider a family of Lebesgue measures $(\nu^*_{x_0})_{x_0\in G_0}$ on the fibers $(\tau_G^*)^{-1}(x_0)$.  Then the fiberwise Fourier transform of a function $f: \mathfrak{G}^*\to \mathbb{C}$ is defined by
\begin{align}
    \hat{f}(X) = \int_{(\tau_G^*)^{-1}(x_0)} d\nu^*_{x_0}(\theta)e^{i\theta(X)} f(\theta)
\end{align}
for all $X\in \mathfrak{G}$.

Next, define $\mathcal{P}_G = C^\infty_{PW}(\mathfrak{G}^*)$, the algebra of smooth Paley-Wiener functions on $\mathfrak{G}^*$, i.e., those functions whose fiberwise Fourier transform is smooth and compactly supported.   Take $\mathfrak{A}_0 = C_0(\mathfrak{G}^*)$ and $\mathfrak{A}_\hbar = C^*(G)$ for all $\hbar >0$.  For each $\hbar>0$, define the quantization map $\mathcal{Q}^G_\hbar: C_{PW}^\infty(\mathfrak{G}^*)\to C^*(G)$ by
\begin{align}
\label{eq:quantization_map}
    \mathcal{Q}^G_\hbar(f)(\exp^W (X)) = \hbar^{-n}\kappa_G(X)\hat{f}(X/\hbar),
\end{align}
for $n = \dim G$ and where $\kappa_G: \mathfrak{G}\to \mathbb{C}$ is an appropriate real-valued smooth cut-off function.  We require that $\kappa_G(-X) = \kappa_G(X)$, and that $\kappa_G$ has support contained in a neighborhood of the zero section in $\mathfrak{G}$ on which $\exp^L$ and $\exp^W$ are diffeomorphisms.  \citet{La99} shows that the maps $\mathcal{Q}_\hbar^G$ satisfy all the conditions of a strict deformation quantization.

Two strict quantizations $(\mathcal{Q}_\hbar)_{\hbar\in I}$ and $(\mathcal{Q}'_\hbar)_{\hbar\in I}$ of a Poisson algebra $\mathcal{P}$ with the same fiber C*-algebras are called \emph{equivalent} if for every $f\in\mathcal{P}$, the map
\begin{align}
    \hbar\mapsto \norm{\mathcal{Q}_\hbar(f)-\mathcal{Q}'_\hbar(f)}
\end{align}
is continuous \citep{La98b}.  For example, it is known that Weyl quantization and Berezin quantization on $\mathbb{R}^{2n}$ are equivalent \citep{La98b}.  Equivalent quantizations share the structure of a continuous bundle of C*-algebras \citep{Di77,FeDo88,KiWa95,Ni96}, which plays a key role in what follows.

We now define the type of bundles of C*-algebras we will employ in this paper, which involve strong properties of \emph{uniform continuity} \citep{StFe21,FeSt24}.  This definition is motivated by that of \citet{KiWa95}, but allows for the assimilation of not only strict deformation quantizations over $\hbar \in [0,1]$, but also their restrictions to $\hbar>0$, i.e., $\hbar\in (0,1]$.
  Our definition of uniformly continuous bundles of C*-algebras provides a structure with enough information to reconstruct the classical limit.\bigskip

\begin{definition} \label{def:bun}
    A \emph{uniformly continuous bundle of C*-algebras} consists in:
    \begin{itemize}
        \item a locally compact metric space $I$ called the \emph{base space};
        \item a family of C*-algebras $(\mathfrak{A}_\hbar)_{\hbar\in I}$ called the \emph{fibers};
        \item a C*-algebra $\mathfrak{A}$ called the algebra of \emph{uniformly continuous sections}; and
        \item a family of surjective *-homomorphisms $(\phi_\hbar: \mathfrak{A}\to\mathfrak{A}_\hbar)_{\hbar\in I}$ called the \emph{evaluation maps}
    \end{itemize}
    satisfying for all $a\in\mathfrak{A}$:
    \begin{enumerate}
        \item $\norm{a} = \sup_{\hbar\in I} \norm{\phi_\hbar(a)}$;
        \item for each uniformly continuous bounded function $f$ on $I$, there is an element $fa\in \mathfrak{A}$ such that $\phi_\hbar(fa) = f(\hbar)\phi_\hbar(a)$ for all $\hbar\in I$; and
        \item the map $\hbar\mapsto \norm{\phi_\hbar(a)}$ is uniformly continuous and bounded.
    \end{enumerate}
\end{definition}\bigskip

\noindent All bundles considered in this paper are uniformly continuous, so we will sometimes drop the modifier ``uniformly" and simply refer to these structures as continuous bundles of C*-algebras.  However, we emphasize that our definition differs from the standard one for continuous bundles of C*-algebras in \citet{KiWa95}, which replaces the role of the uniformly continuous functions on $I$ in Def. \ref{def:bun} with merely continuous functions that vanish at infinity on $I$.  We also note that we will sometimes refer to a uniformly continuous bundle of C*-algebras by the algebra of continuous sections $\mathfrak{A}$ for shorthand.

A strict quantization determines a continuous bundle of C*-algebras as long as it satisfies mild technical conditions.  The following construction produces a continuous bundle of C*-algebras just in case for every polynomial $P$ of the maps $[\hbar\mapsto \mathcal{Q}_\hbar(f)]$ for $f\in\mathcal{P}$, the map
\begin{align}
    \hbar\mapsto \norm{P(\hbar)}
\end{align}
is continuous \citep{BiHoRi04b}.  One generates a *-algebra $\tilde{\mathfrak{A}}\subseteq\prod_{\hbar\in [0,1]}$ of sections by taking pointwise products and sums of maps of the form
\begin{align}
\label{eq:quantization_sections}
    [\hbar\mapsto \mathcal{Q}_\hbar(f)]
\end{align}
for $f\in\mathcal{P}$.  Then one defines a C*-algebra of continuous sections by
\begin{align}
\label{eq:generating_sections}
    \mathfrak{A} = \left\{a\in\prod_{\hbar\in [0,1]}\mathfrak{A}_\hbar\ \Bigg|\ \text{ the map } \left[\hbar\mapsto\norm{a(\hbar)-\tilde{a}(\hbar)}\right]\text{ is continuous for each $\tilde{a}\in\tilde{\mathfrak{A}}$}\right\}.
\end{align}
With the evaluation maps given by
\begin{align}
    \phi_\hbar(a) = a(\hbar)
\end{align}
for each $a\in\mathfrak{A}$ and $\hbar\in [0,1]$, this structure becomes a continuous bundle of C*-algebras.  Indeed, one can show that equivalent quantizations determine the same algebra of continuous sections, and hence the same continuous bundle \citep{La98b}.

We are primarily interested in the case where one has a strict quantization over the base space $[0,1]$, which determines a continuous bundle of C*-algebras over $[0,1]$.  Then the restriction of a continuous bundle over the closed interval $[0,1]$ to the half open interval $(0,1]$ is a uniformly continuous bundle of C*-algebras.  We understand a uniformly continuous bundle of C*-algebras over $(0,1]$ to encode information about the algebras of quantum observables and their scaling properties for values $\hbar>0$.

Using this structure, we can consider the classical $\hbar\to 0$ limit as the ``inverse" process to quantization.  We will represent the $\hbar\to 0$ limit as the extension of a uniformly continuous bundle of C*-algebras over $(0,1]$ to the closed interval $[0,1]$, which includes the value $\hbar=0$.

To that end, suppose we are given a uniformly continuous bundle of C*-algebras over $(0,1]$ with the algebra of uniformly continuous sections $\mathfrak{A}$, fibers $(\mathfrak{A}_\hbar)_{\hbar\in (0,1]}$, and evaluation maps $(\phi_\hbar)_{\hbar\in (0,1]}$.  \citet{StFe21a} show that whenever the inclusion of base spaces $(0,1]\hookrightarrow [0,1]$ is isometric, which clearly holds in this case, then the bundle can be uniquely extended to the one-point compactification $[0,1]$ with the fiber algebra at any $\hbar\in [0,1]$ (including the value $\hbar =0$) determined by
\begin{align}
    \mathfrak{A}_\hbar = \mathfrak{A}/K_\hbar,
\end{align}
where one takes the quotient by the closed two-sided ideal
\begin{align}
\label{eq:ideal}
    K_\hbar = \left\{a\in\mathfrak{A}\ |\ \lim_{\hbar'\to \hbar}\norm{\phi_{\hbar'(a)}} = 0\right\}.
\end{align}
In particular, one has the unique limiting fiber algebra $\mathfrak{A}_0 = \mathfrak{A}/K_0$.

When the uniformly continuous bundle is determined by a strict deformation quantization of a Poisson algebra, we have that $\mathfrak{A}_1$ represents the fully quantum theory and $\mathfrak{A}_0$ represents the classical theory.  We mention briefly some features of the classical limit that are described in more detail by \citet{StFe21a} and \citet{FeSt24}. First, von Neumann's condition implies that the classical limit algebra $\mathfrak{A}_0$ is an abelian algebra, and hence is an algebra of functions on a topological space $M$ by the Gelfand representation theorem.  Second, certain choices of Poisson algebras carry enough information to define the structure of a smooth manifold on $M$. Third, Dirac's condition implies that the bundle of C*-algebras determines a Poisson bracket on $M$.  Finally, the classical limit is unique up to isomorphism in the sense that the quantization of a Poisson algebra of functions on $M$ recovers the structure of $M$ in the classical limit, with $\mathfrak{A}_0\cong C_0(M)$.  To summarize: the classical limit of a uniformly continuous bundle of C*-algebras generated from a strict deformation quantization allows one to recover the geometric structures of the original space, and thus to understand the classical structures as being determined by the quantum model for $\hbar>0$.

\section{Classical and Quantum Morphisms}
\label{sec:arrows}

In order to understand quantization as a categorical equivalence between models of classical and quantum physics, we must specify the morphisms for the categories of models that we consider.  Following \citet{La01a,La01c}, we treat morphisms that capture the representation-theoretic structure of classical and quantum models, respectively. 
 We now proceed to define morphisms between classical models (i.e., Poisson manifolds) as symplectic dual pairs, and morphisms between quantum models (i.e., C*-algebras) as Hilbert bimodules.  We will focus on symplectic dual pairs and Hilbert bimodules that arise from bibundles between Lie groupoids, which consist in pairs of groupoids acting on a common space \citep{Ma87,Mo88,MuReWi87}.  We first provide the definition of a bibundle that we will work with.\bigskip

 \begin{definition}
     A \emph{bibundle} between Lie groupoids $G$ and $H$ is a manifold $M$ with smooth maps $t_M: M\to G_0$ and $s_M: M\to H_0$, carrying a smooth left $G$-action and a smooth right $H$-action.  For $x\in G$ and $q\in M$, the result of the left $G$-action $xq$ is defined just in case $s_G(x) = t_M(q)$.  For $y\in H$ and $q\in M$, the result of the right $G$-action $qy$ is defined just in case $s_M(q) = t_H(y)$.  Moreover, the bibundle denoted $G\rightarrowtail M\leftarrowtail H$ is required to satisfy $t_M(qy) = t_M(q)$, $s_M(xq) = s_M(q)$, and $(xq)y = x(qy)$ for all $x\in G$, $q\in M$, and $y\in H$ on which the groupoid actions are defined.
 \end{definition}\bigskip

We will sometimes refer to a bibundle $G\rightarrowtail M \leftarrowtail H$ by its middle space $M$.  Two bibundles $G\rightarrowtail M\leftarrowtail H$ and $G\rightarrowtail N\leftarrowtail H$ are called \emph{isomorphic} if there is a diffeomorphism $\zeta: M\to N$ such that $s_N\circ \zeta = s_M$, $t_N\circ \zeta = t_M$, and for all $x\in G$, $q\in M$, and $y\in H$, $\zeta(xq) = x\zeta(q)$ and $\zeta(qy) = \zeta(q)y$.  A bibundle is called \emph{regular} when
\begin{itemize}
    \item $s_M$ is a surjective submersion;
    \item for all $x\in G$ and $q\in M$, $xq = q$ iff $x\in G_0$;
    \item for all $q\in M$ and $q'\in s_M^{-1}(s_M(q))$, there is an $x\in G$ such that $q' = xq$;
    \item the right $H$-action is proper.
\end{itemize}
Bibundles that are regular can be thought of as generalized maps from $H$ to $G$, and there is a natural composition operator that can be defined via a tensor product \citep{La01b}.  In particular, given two regular bibundles
\begin{align}
    G\rightarrowtail M\leftarrowtail H && \text{and} && H\rightarrowtail N\leftarrowtail G',
\end{align}
define the fiber product
\begin{align}
    M\times_H N = \Big\{(q,q')\in M\times N\ \Big|\ s_M(q) = t_N(q')\Big\}.
\end{align}
This space carries an $H$-action defined for $y\in H$ by $(q,q')\in M\times_H N\mapsto (qy,y^{-1}q)$.  The tensor product, defined as $M\circledast_H N = (M\times_H N)/H$
forms the middle space of a bibundle
\begin{align}
    G\rightarrowtail M\circledast_H N\leftarrowtail G',
\end{align}
which we consider as the tensor product of the bibundles $M$ and $N$ \citep{La01b,La01c}.
\subsection{Dual Pairs and Hilbert Bimodules}

We follow \citet{We83} and \citet{La01b} in the definition of symplectic duals pairs (see also \citet{Xu91,Xu92,Xu94}).\bigskip

\begin{definition}
    A \emph{symplectic dual pair} $M\leftarrow S\rightarrow N$ between Poisson manifolds $M$ and $N$ consists in a symplectic manifold $S$ and smooth Poisson maps $j_M: S\to M$ and $j_N: S\to N$, where $N$ is given minus the Poisson bracket.  Together, they must satisfy for every $f\in C^\infty(M)$ and $g\in C^\infty(N)$,
    \begin{align}
        \{j^*_Mf,j^*_Ng\} = 0,
    \end{align}
    where the Poisson bracket is determined by the symplectic form on $S$.
\end{definition}\bigskip

Sometimes we will refer to the dual pair $M\leftarrow S\rightarrow N$ by its middle space $S$ for shorthand.  We call another dual pair $M\leftarrow \tilde{S}\rightarrow N$ with maps $\tilde{j}_M: \tilde{S}\to M$ and $\tilde{j}_N: S\to N$ \emph{isomorphic} to the dual pair $M\leftarrow S\rightarrow N$ when there is a symplectomorphism $u: S\to \tilde{S}$ such that $\tilde{j}_M\circ u = j_M$ and $\tilde{j}_N\circ u = j_N$.

We call a symplectic dual pair \emph{weakly regular} if the maps $j_M$ and $j_N$ are complete, and $j_N$ is a surjective submersion \citep{La01b}.  Here, we say $j_M$ is complete if whenever a function $f\in C^\infty(M)$ has complete Hamiltonian flow on $M$, then $j^*_Mf$ has complete Hamiltonian flow on $S$, and likewise for $j_N$ \citep{La98b}.  Symplectic dual pairs that are weakly regular can be composed by a kind of tensor product as follows \citep{La01b}.  Consider two weakly regular symplectic dual pairs
\begin{align}
    \xymatrix{
M & S_1 \ar[l]_{\overset{1}{j}_M}\ar[r]^{\overset{1}{j}_N} & N
} && \text{and} && \xymatrix{
N & S_2 \ar[l]_{\overset{2}{j}_N}\ar[r]^{\overset{2}{j}_P} & P
}.
\end{align}
Define the space
\begin{align}
    S_1\times_N S_2 = \Big\{(s_1,s_2)\in S_1\times S_2\ \Big|\ \overset{1}{j}_N(s_1) = \overset{2}{j}_N(s_2)\Big\}.
\end{align}
Since $S_1\times S_2$ carries the product symplectic structure determined by the symplectic manifolds $S_1$ and $S_2$, we can consider the null distribution $\mathcal{N}_{S_1\times_N S_2}$ of vectors in $T(S_1\times_N S_2)$ that are null relative to the symplectic form on the product space.  Then define $S_1\circledcirc_NS_2 = (S_1\times_N S_2)/\mathcal{N}_{S_1\times_NS_2}$.  Weak regularity ensures this is a symplectic manifold that serves as the middle space of a symplectic dual pair
\begin{align}
    \xymatrix{
M & S_1\circledcirc_NS_2 \ar[l]\ar[r] & P
},
\end{align}
which we consider as the composition of the dual pairs $S_1$ and $S_2$ \citep{La02a}.

It is well known that a symplectic dual pair $M\leftarrow S\rightarrow N$ allows one to induce symplectic realizations of $M$ from those of $N$ \citep{La95b}.  Here, a \emph{symplectic realization} is a smooth Poisson map $j: S_0\to N$ from a symplectic manifold $S$ to a Poisson manifold $N$.  One can induce a symplectic realization of $M$ from $S_0$ and the dual pair $S$ by taking the above tensor product composition $S\circledcirc S_0$ with the obvious map $\tilde{j}: S\circledcirc S_0\to M$ constructed from $j_M$ on the first component.  Moreover, this construction of induced symplectic realizations does not depend on the full structure of the symplectic dual pair $S$, but only on the image $(j_M,j_N)[S]\subseteq M\times N$, which is a \emph{Lagrangian relation} \citep{We10}.  Here, $(j_M,j_N): S\to M\times N$ is the pair map $(j_M,j_N)(s) = (j_M(s),j_N(s))$ for $s\in S$.  By saying the induced symplectic realization only depends on the Lagrangian relation, we mean that distinct symplectic dual pairs may give rise to equivalent functors from realizations of $N$ to realizations of $M$, but such dual pairs will determine the same Lagrangian relation.  Moreover, any symplectic dual pairs that determine the same Lagrangian relation will also determine the same induction functor.  Hence, we will define the morphisms in our category as equivalence classes of symplectic dual pairs that determine the same Lagrangian relation.

The main example of dual pairs we consider in this paper arises from a regular bibundle $G\rightarrowtail M\leftarrowtail H$ between Lie groupoids $G$ and $H$, which allows one to construct a symplectic dual pair $\mathfrak{G}^*\leftarrow T^*M\rightarrow \mathfrak{H}^*$ between the corresponding duals of their Lie algebras $\mathfrak{G}^*$ and $\mathfrak{H}^*$.  The Poisson morphism $j_G: T^*M\to \mathfrak{G}^*$ is defined by
\begin{align}
\label{eq:left_momentum_map}
    j_G(\eta_q)\Big(\frac{d}{dt}_{|t=0} \gamma(t)\Big) = -\eta_q\Big(\frac{d}{dt}_{|t=0}\big(\gamma(t)^{-1}\cdot q\big)\Big),
\end{align}
where $\eta_q\in T^*_qM$ and $\gamma:I\to G$ is a curve with $\gamma(0) = t_M(q)$ and $\gamma(t)\in t_G^{-1}(t_M(q))$ so that $\dot{\gamma}(0)\in \tau_G^{-1}(t_M(q))$ and $\gamma(t)^{-1}\cdot q$ is well-defined for all $t\in I\subseteq \mathbb{R}$.  Likewise, the Poisson morphism $j_H: T^*M\to \mathfrak{H}^*$ is defined by
\begin{align}
\label{eq:right_momentum_map}
    j_H(\eta_q)\Big(\frac{d}{dt}_{|t=0} \gamma(t)\Big) = \eta_q\Big(\frac{d}{dt}_{|t=0} \big(q\cdot \gamma(t)\big)\Big),
\end{align}
where $\eta_q\in T^*_qM$ and $\gamma: I\to H$ is a curve with $\gamma(0) = s_M(q)$ and $\gamma(t)\in t_H^{-1}(s_M(q))$, so that $\dot{\gamma}(0)\in \tau_H^{-1}(s_M(q))$ and $q\cdot \gamma(t)$ is well-defined for all $t\in I\subseteq \mathbb{R}$.  It follows that $\mathfrak{G}^*\leftarrow T^*M\rightarrow \mathfrak{H}^*$ satisfies the definition of a symplectic dual pair \citep{La00,La01a}.  Regularity of bibundles implies weak regularity of the corresponding dual pairs, which guarantees that these dual pairs can be composed by means of the tensor product.

We will further restrict attention to integrable Poisson manifolds (or Lie groupoids for which the dual of the corresponding Lie algebroid is integrable), since the condition of integrability guarantees that there is a symplectic dual pair from the Poisson manifold to itself that serves the role of an identity arrow \citep{La01b,CrFe04}.  Similarly, we restrict attention to source-simply connected Lie groupoids (i.e., Lie groupoids $G$ for which $s_G^{-1}(x_0)$ is connected and simply connected for each $x_0\in G_0$ \citep{MoMr02}) because this condition holds just in case the identity bibundle corresponds to the identity dual pair.

We define the following category of Poisson manifolds, understood as a category of classical models.\bigskip

\begin{definition}
    We denote by $\mathbf{LPoisson}$ the category consisting of the following:
    \begin{itemize}
        \item \emph{Objects}: integrable duals of Lie algebroids $\mathfrak{G}^*$ associated to source-simply connected Lie groupoids $G$, understood as a class of Poisson manifolds.
        \item \emph{Arrows}: equivalence classes of symplectic dual pairs $\mathfrak{G}^*\leftarrow T^*M\rightarrow\mathfrak{H}^*$ associated to regular bibundles $M$ between Lie groupoids via the construction of Eqs. (\ref{eq:left_momentum_map})-(\ref{eq:right_momentum_map}) that determine the same Lagrangian relation $(j_G,j_H)[T^*M]\subseteq \mathfrak{G}^*\times\mathfrak{H}^*$.
    \end{itemize}
\end{definition}

\noindent $\mathbf{LPoisson}$ is indeed a category: integrability of the objects guarantees the existence of identity arrows, and the behavior of the tensor product for dual pairs implies the appropriate composition laws for arrows \citep{La01b,La01c}.  Note that our definition differs from the one given by Landsman because we consider equivalence classes of symplectic dual pairs that determine the same Lagrangian relation rather than isomorphism classes of symplectic dual pairs themselves.

Next, we follow \citet{La95} and \citet{RaWi98} in the definition of Hilbert bimodules (see also \citet{Ri74,Ri74a,Ri72a}).\bigskip

\begin{definition}
A \emph{Hilbert bimodule} $\mathfrak{A}\rightarrowtail\mathcal{E}\leftarrowtail\mathfrak{B}$ between C*-algebras $\mathfrak{A}$ and $\mathfrak{B}$ consists in a vector space $\mathcal{E}$ carrying a $\mathfrak{B}$-valued positive definite inner product $\inner{\cdot}{\cdot}_\mathcal{E}$ that is antilinear in the first argument and linear in the second.  The space $\mathcal{E}$ carries a right action of $\mathfrak{B}$ and a left action of $\mathfrak{A}$ by adjointable operators (Here, a linear operator $a$ on $\mathcal{E}$ is called \emph{adjointable} if there is a linear operator $a^*$ such that $\inner{\varphi_1}{a\varphi_2}_{\mathcal{E}} = \inner{a^*\varphi_1}{\varphi_2}_{\mathcal{E}}$ for all $\varphi_1,\varphi_2\in\mathcal{E}$.)  The inner product is required to satisfy
\begin{align}
    \inner{\varphi_1}{\varphi_2b}_\mathcal{E} = \inner{\varphi_1}{\varphi_2}_\mathcal{E}\cdot b
\end{align}
for all $\varphi_1,\varphi_2\in\mathcal{E}$ and $b\in\mathfrak{B}$, and the space $\mathcal{E}$ is required to be complete in the norm
\begin{align}
    \norm{\varphi}^2 = \norm{\inner{\varphi}{\varphi}_{\mathcal{E}}},
\end{align}
where the right hand side uses the C*-norm in $\mathfrak{B}$.  Finally, the $\mathfrak{A}$-action is required to be nondegenerate in the sense that $\mathfrak{A}\mathcal{E}$ is dense in $\mathcal{E}$.
\end{definition}\bigskip

\noindent Sometimes we will refer to the bimodule $\mathfrak{A}\rightarrowtail\mathcal{E}\leftarrowtail\mathfrak{B}$ by its middle space $\mathcal{E}$ for shorthand.  We call another Hilbert bimodule $\mathfrak{A}\rightarrowtail\mathcal{F}\leftarrowtail\mathfrak{B}$ \emph{unitarily equivalent} to the Hilbert bimodule $\mathfrak{A}\rightarrowtail\mathcal{E}\leftarrowtail\mathfrak{B}$ when there is a unitary map $U: \mathcal{E}\to\mathcal{F}$ intertwining both the left $\mathfrak{A}$-actions and the right $\mathfrak{B}$-actions.

 One can think of Hilbert bimodules as generalized maps between C*-algebras, which can be composed by means of a tensor product operation, as follows \citep{La01b,La01c}.  Consider two Hilbert bimodules
 \begin{align}
\mathfrak{A}\rightarrowtail\mathcal{E}\leftarrowtail\mathfrak{B} && \text{and} && \mathfrak{B}\rightarrowtail\mathcal{F}\leftarrowtail\mathfrak{C}.
 \end{align}
Consider the algebraic tensor product $\mathcal{E}\otimes \mathcal{F}$ and define the subspace
\begin{align}
    \mathcal{N}_\mathfrak{B} = \text{span}\Big\{\varphi b\otimes \psi - \varphi\otimes b\psi\ \Big|\ \varphi\in\mathcal{E}, \psi\in\mathcal{F}, \text{and }b\in\mathfrak{B}\Big\}.
\end{align}
Then define $\mathcal{E}\dot{\otimes}_{\mathfrak{B}} \mathcal{F} = (\mathcal{E}\otimes\mathcal{F})/\mathcal{N}_\mathfrak{B}$.  This space carries a natural left $\mathfrak{A}$-action and right $\mathfrak{C}$-action as well as the $\mathfrak{C}$-valued inner product defined by
\begin{align}
    \inner{\varphi_1\otimes \psi_1}{\varphi_2\otimes\psi_2}_{\mathcal{E}\dot{\otimes}_\mathfrak{B}\mathcal{F}} = \inner{\psi_1}{\inner{\varphi_1}{\varphi_2}_{\mathcal{E}}\cdot\psi_2}_{\mathcal{F}}
\end{align}
for $\varphi_1,\varphi_2\in\mathcal{E}$ and $\psi_1,\psi_2\in\mathcal{F}$.  The completion of $\mathcal{E}\dot{\otimes}_{\mathfrak{B}}\mathcal{F}$, denoted $\mathcal{E}\otimes_\mathfrak{B}\mathcal{F}$, is called the interior tensor product and forms the middle space of a Hilbert bimodule
\begin{align}
    \mathfrak{A}\rightarrowtail\mathcal{E}\otimes_\mathfrak{B}\mathcal{F}\leftarrowtail\mathfrak{C},
\end{align}
which we consider as the composition of the Hilbert bimodules $\mathcal{E}$ and $\mathcal{F}$.

It is well known that a Hilbert bimodule $\mathfrak{A}\rightarrowtail\mathcal{E}\leftarrowtail\mathfrak{B}$ allows one to induce Hilbert space representations of $\mathfrak{A}$ from those of $\mathfrak{B}$ \citep{La95b} (See also \citep{Ri72a,Ri74,Ri74a,Gr78,Gr80,Ro94a}).  Here, a \emph{Hilbert space representation} is a *-homomorphism $\pi: \mathfrak{B}\to \mathcal{H}_\pi$.  One can induce a Hilbert space representation of $\mathfrak{A}$ from $\mathcal{H}_\pi$ and the bimodule $\mathcal{E}$ by taking the above interior tensor product $\mathcal{E}\otimes_\mathfrak{B}\mathcal{H}_\pi$.  Moreover, this construction of induced Hilbert space representations does not depend on the full structure of the Hilbert bimodule $\mathcal{E}$ in the sense that distinct Hilbert bimodules can give rise to the same induction functor.  But there is a structure that is invariant among the different Hilbert bimodules giving rise to the same induction functor.  If we denote by $\pi_\mathfrak{A}^\mathcal{E}$ and $\pi_\mathfrak{B}^\mathcal{E}$ the representations of $\mathfrak{A}$ and (the opposite of) $\mathfrak{B}$ determined by their left/right actions on $\mathcal{E}$, then the induced Hilbert space representation possesses an invariant in the form of the kernel $\ker(\pi_\mathfrak{A}^\mathcal{E}\otimes\pi_\mathfrak{B}^\mathcal{E})$ of the tensor product representation, which we will henceforth call the \emph{tensor product kernel} associated with $\mathcal{E}$.  By calling the tensor product kernel an invariant of the induced representation, we mean that while distinct Hilbert bimodules may give rise to equivalent functors from Hilbert space representations of $\mathfrak{B}$ to Hilbert space representations of $\mathfrak{A}$, in this case the Hilbert bimodules must possess the same tensor product kernel.  Hence, we will define the morphisms in our category as equivalence classes of Hilbert bimodules that share the same tensor product kernel.

In particular, a regular bibundle $G\rightarrowtail M\leftarrowtail H$ between Lie groupoids $G$ and $H$ allows one to construct a Hilbert bimodule $C^*(G)\rightarrowtail\mathcal{E}_M\leftarrowtail C^*(H)$, where $\mathcal{E}_M$ is the completion of $C_c^\infty(M)$ relative to the inner product defined below.  To define the inner product and group actions, first note that the left Haar system $(\nu^t_{x_0})_{x_0\in G_0}$ on $G$ induces a family of measures $(\mu_{y_0})_{y_0\in H_0}$ on $M$, each with support contained in $s_M^{-1}(y_0)$, such that for each $f\in C_c^\infty(M)$ \citep{La01a},
\begin{enumerate}
    \item the function $y_0\mapsto \int d\mu_{y_0}(q)\ f(q)$ on $H_0$ is smooth;
    \item $\int d\mu_{t_H(y)}(q)\ f(qy) = \int d\mu_{s_H(y)}(q)f(q)$; and
    \item for any $y_0\in H_0$ and $q_0\in s_M^{-1}(y_0)$, $\int d\mu_{y_0}(q)\ f(q) = \int d\nu^t_{t_M(q_0)}(x)\ f(x^{-1}q_0)$.
\end{enumerate}
We will further denote the left Haar system on $H$ by $(\lambda^t_{y_0})_{y_0\in H_0}$.
With the measures $(\mu_{y_0})_{y_0\in H_0}$ on $M$, we can define the $C^*(H)$-valued inner product
\begin{align}
\label{eq:inner_product}
    \inner{\varphi}{\psi}_{\mathcal{E}_M}(y) = \int_{s_M^{-1}(t_H(y))}d\mu_{t_H(y))}(q)\  \overline{\varphi(q)}\psi(qy)
\end{align}
for $\varphi,\psi\in C_c^\infty(M)$ and $y\in H$.  Then $\mathcal{E}_M$ is defined as the completion of $C_c^\infty(M)$ relative to the norm $\norm{\varphi}_{\mathcal{E}_M}^2 = \norm{\inner{\varphi}{\varphi}_{\mathcal{E}_M}}$, where the norm on the right hand side is the C*-norm in $C^*(H)$.  The left $C^*(G)$-action and right $C^*(H)$ action are given for $f\in C^*(G)$ and $g\in C^*(H)$ by
\begin{align}
\label{eq:left_bimodule_action}
    (f\cdot \varphi)(q) &= \int_{t_G^{-1}(t_M(q))}d\nu^t_{t_M(q)}(x)\ f(x)\varphi(x^{-1}q)\\
\label{eq:right_bimodule_action}(\varphi\cdot g)(q) &= \int_{t_H^{-1}(s_M(q))} d\lambda^t_{s_M(q)}(y)\ g(y^{-1})\varphi(qy)
\end{align}
for all $\varphi\in C_c^\infty(M)$ and $q\in M$, extended to $\mathcal{E}_M$ by continuity. It then follows that $C^*(G)\rightarrowtail \mathcal{E}_M\leftarrowtail C^*(H)$ satisfies the definition of a Hilbert bimodule \citep{Gr78,MuReWi87,La01a}.

We define the following category of C*-algebras, understood as a category of quantum models.\bigskip

\begin{definition}
    We denote by $\mathbf{LC^*}$ the category consisting in:
    \begin{itemize}
        \item \emph{Objects}: Lie groupoid C*-algebras $C^*(G)$ associated to source-simply connected Lie groupoids $G$ whose Lie algebroids $\mathfrak{G}^*$ are integrable.
        \item \emph{Arrows}: equivalence classes of Hilbert bimodules $C^*(G)\rightarrowtail \mathcal{E}_M\leftarrowtail C^*(H)$ associated to integrable, regular bibundles $M$ between Lie groupoids via the construction of Eqs. (\ref{eq:inner_product})-(\ref{eq:right_bimodule_action}) that determine the same tensor product kernel $\ker(\pi_\mathfrak{A}^\mathcal{E}\otimes\pi_\mathfrak{B}^\mathcal{E})$.
    \end{itemize}
\end{definition}
\noindent One can check that $\mathbf{LC^*}$ is indeed a category, with the behavior of the tensor product for Hilbert bimodules implying the appropriate composition laws for arrows \citep{La01b,La01c}.  Again, note that our definition differs from the one given by Landsman because we consider equivalence classes of Hilbert bimodules that determine the same tensor product kernel rather than unitary equivalence classes of Hilbert bimodules themselves. 

\subsection{Quantization and the Classical Limit of Morphisms}

Quantization and the classical limit already provide ways to associate objects in $\mathbf{LPoisson}$ with objects in $\mathbf{LC^*}$.  In particular, the strict quantization defined in Eq. (\ref{eq:quantization_map}) of the dual of a Lie algebroid $\mathfrak{G}^*$ in $\mathbf{LPoisson}$ associated to a Lie groupoid $G$ yields the Lie groupoid C*-algebra $C^*(G)$ in $\mathbf{LC^*}$.  On the other hand, one can take the classical limit of an object $C^*(G)$ in $\mathbf{LC^*}$ by generating a uniformly continuous bundle $\mathfrak{A}^G$ over the base space consisting in values of $\hbar$ in $(0,1]$ from this strict quantization. (\citet{La99} shows the algebra of continuous sections for this uniformly continuous bundle is the groupoid C*-algebra of the normal groupoid over $G$.)  Let $K^G_0$ be the ideal of sections in $\mathfrak{A}^G$ whose norm vanishes at $\hbar = 0$ as in Eq. (\ref{eq:ideal}).  The construction of \citet{StFe21a} then yields the classical limit $\mathfrak{A}^G/K^G_0$, which they show is unique up to isomorphism.  This implies that $\mathfrak{A}^G/K^G_0\cong C_0(\mathfrak{G}^*)$, which leads by Gelfand duality to the Poisson manifold $\mathfrak{G}^*$ in $\mathbf{LPoisson}$.

We will denote these associations by
\begin{align}
    &Q: \mathfrak{G}^*\mapsto C^*(G)\\
    &L: C^*(G)\mapsto \mathfrak{G}^*.
\end{align}
The map $Q$ represents the quantization of objects in $\mathbf{LPoisson}$ yielding objects in $\mathbf{LC^*}$.  The map $L$ represents the classical limit of objects in $\mathbf{LC^*}$ yielding objects in $\mathbf{LPoisson}$.

We will understand each of the maps $Q$ and $L$ as the action of a functor on objects in the respective category.  To complete the definition of each functor, we must specify the behavior of each of $Q$ and $L$ on arrows.  \citet{La01a,La02a} (See also \citep{La00}) concretely defines the action of $Q$ on arrows in $\mathbf{LPoisson}$ via the association
\begin{align}
    Q: [\mathfrak{G}^*\leftarrow T^*M\rightarrow \mathfrak{H}^*]\mapsto [C^*(G)\rightarrowtail\mathcal{E}_M\leftarrowtail C^*(H)]
\end{align}
The fact that $Q$ is well-defined follows from our Prop.~\ref{prop:kernel} below, which shows that if two symplectic dual pairs $T^*M$ and $T^*M'$ determine the same Lagrangian relation, then the Hilbert bimodules $\mathcal{E}_M$ and $\mathcal{E}_{M'}$ must have the same tensor product kernel.  Landsman shows that $Q$ respects the composition of arrows by the tensor product of symplectic dual pairs in $\mathbf{LPoisson}$ and the interior tensor product of Hilbert bimodules in $\mathbf{LC^*}$, up to isomorphism of symplectic dual pairs and unitary equivalence of Hilbert bimodules.  Landsman's result guarantees that $Q$ is indeed a functor.

While the quantization of arrows is specific to the objects of $\mathbf{LPoisson}$ rather than general symplectic dual pairs, the classical limit of arrows as constructed in \citet{FeSt24} applies to a broad class of Hilbert bimodules, of which we will see the arrows in $\mathbf{LC^*}$ as a special case.  They consider the case where one has not only the fiber C*-algebras at some particular value of $\hbar\in (0,1]$, but also where one possesses two full uniformly continuous bundles of C*-algebras $\mathfrak{A}$ and $\mathfrak{A}'$ generated by quantization maps $(\mathcal{Q}_\hbar)_{\hbar\in (0,1]}$ and $(\mathcal{Q}'_\hbar)_{\hbar\in (0,1]}$.  Under certain conditions, a Hilbert bimodule $\mathcal{E}_\hbar$ between the fibers $\mathfrak{A}_\hbar$ and $\mathfrak{A}'_\hbar$ at a fixed value yields a Hilbert bimodule $\underline{\mathcal{E}}$ between the algebras of continuous sections $\mathfrak{A}$ and $\mathfrak{A}'$ by translating between values of $\hbar\in (0,1]$ by means of the maps $\mathcal{Q}$ and $\mathcal{Q}'$.  One can then use the Hilbert bimodule $\mathfrak{A}\rightarrowtail\underline{\mathcal{E}}\leftarrowtail\mathfrak{A}'$ to construct the classical limit as follows.

Let $K_0 = \{a\in\mathfrak{A}\ |\ \lim_{\hbar\to 0}\norm{\phi_{\hbar}(a)} = 0\}$ and $K'_0 = \{b\in\mathfrak{A}'\ |\ \lim_{\hbar\to 0}\norm{\phi_{\hbar}(b)} = 0\}$.  Recall that the classical limit algebras are $\mathfrak{A}_0 = \mathfrak{A}/K_0$ and $\mathfrak{A}'_0 = \mathfrak{A}'/K'_0$.
We will say that the Hilbert bimodule $\underline{\mathcal{E}}$ is \emph{strongly nondegenerate} at $\hbar = 0$ if $\overline{K_0\cdot \underline{\mathcal{E}}}\subseteq \overline{\underline{\mathcal{E}}\cdot K'_0}$.  \citet{FeSt24} show that if $\underline{\mathcal{E}}$ is strongly nondegenerate at $\hbar= 0$, then
\begin{align}
    \mathfrak{A}_0\rightarrowtail\mathcal{E}_0\leftarrowtail\mathfrak{A}_0'
\end{align}
forms a Hilbert bimodule between the classical limit fiber C*-algebras  $\mathfrak{A}_0$ and $\mathfrak{A}_0'$ at $\hbar = 0$, with the middle space given by $\mathcal{E}_0 = \underline{\mathcal{E}}/\overline{\underline{\mathcal{E}}\cdot K'_0}$.  They call the Hilbert bimodule $\mathcal{E}_0$ the \emph{Hilbert classical limit} of $\underline{\mathcal{E}}$.

One can then use a duality to construct a symplectic dual pair from the Hilbert classical limit $\mathcal{E}_0$.  In particular the left $\mathfrak{A}_0$-action and right $\mathfrak{A}_0$-action each give rise to representations $\pi_0$ of $\mathfrak{A}_0$ and $\pi'_0$ of (the opposite of) $\mathfrak{A}'_0$ as bounded linear operators on $\mathcal{E}_0$.  \citet{FeSt24} consider the C*-algebra
\begin{align}
    \mathfrak{B} = (\mathfrak{A}_0\otimes\mathfrak{A}'_0)/\ker(\pi_0\otimes \pi'_0),
\end{align}
and show that under certain conditions, $\mathfrak{B}\cong C_0(S)$ for a symplectic manifold $S$.  Moreover, if $\mathfrak{A}$ and $\mathfrak{A}'$ are generated from strict quantizations, then $\mathfrak{A}_0\cong C_0(P)$ and $\mathfrak{A}'_0\cong C_0(P')$ for Poisson manifolds $P$ and $P'$.  In this case, the results of \citet{FeSt24} imply that
\begin{align}
   \xymatrix{P & S\ar[l]\ar[r] & P'}
\end{align}
is a symplectic dual pair, which they call the \emph{symplectic classical limit}.  Importantly, \citet{FeSt24} establish that the association of $\underline{\mathcal{E}}$ with both the Hilbert classical limit $\mathcal{E}_0$ and the symplectic classical limit $S$ respects composition of Hilbert bimodules and symplectic dual pairs by their corresponding tensor product operations, which implies that the association is functorial. 

Let us now return to the case of interest in this paper, where we consider the quantizations $(\mathcal{Q}_\hbar^G)_{\hbar\in (0,1]}$ and $(\mathcal{Q}_\hbar^H)_{\hbar\in (0,1]}$ for Lie groupoids $G$ and $H$ as generating uniformly continuous bundles with C*-algebras of uniformly continuous sections given by $\mathfrak{A}^G$ and $\mathfrak{A}^H$ (i.e., the normal groupoid C*-algebras), respectively.  In the next section, we will show that each Hilbert bimodule $C^*(G)\rightarrowtail\mathcal{E}_M\leftarrowtail C^*(H)$ in $\mathbf{LC^*}$ can be canonically (and functorially) lifted to a Hilbert bimodule \begin{align}    
\mathfrak{A}^G\rightarrowtail\underline{\mathcal{E}}_M\leftarrowtail \mathfrak{A}^H
\end{align}
between the C*-algebras of uniformly continuous sections.  With 
\begin{align}
    K^G_0 = \left\{a\in\mathfrak{A}^G\ |\ \lim_{\hbar\to 0}\norm{a(\hbar)} = 0 \right\} && \text{and} && K^H_0 = \left\{b\in\mathfrak{A}^H\ |\ \lim_{\hbar\to 0}\norm{b(\hbar)} = 0 \right\},
\end{align} the construction of \citet{FeSt24} then produces the Hilbert classical limit
\begin{align}
\mathfrak{A}^G_0\rightarrowtail(\mathcal{E}_M)_0\leftarrowtail\mathfrak{A}^H_0
\end{align}
with middle space $(\mathcal{E}_M)_0 = \underline{\mathcal{E}}_M/\overline{\underline{\mathcal{E}_M}\cdot K^H_0}$ and actions of the abelian C*-algebras $\mathfrak{A}^G_0\cong C_0(\mathfrak{G}^*)$ and $\mathfrak{A}^H_0\cong C_0(\mathfrak{H}^*)$ at $\hbar=0$.  As before, the left $\mathfrak{A}^G_0$-action and right $\mathfrak{A}^H_0$-action on $(\mathcal{E}_M)_0$ give rise to representations $\pi^G_0$ of $\mathfrak{A}^G_0$ and $\pi^H_0$ of (the opposite of) $\mathfrak{A}^H_0$ as bounded linear operators.  This allows us to consider the C*-algebra 
\begin{align}
\label{eq:symplectic_classical_limit}\mathfrak{B}_M =(\mathfrak{A}^G_0\otimes\mathfrak{A}^H_0)/\ker(\pi^G_0\otimes\pi^H_0)
\end{align}
Since $\mathfrak{A}^G_0$ and $\mathfrak{A}^H_0$ are abelian, so is $\mathfrak{B}_M$.  Gelfand duality then implies $\mathfrak{B}_M\cong C_0(S_M)$ for a topological space $S_M$, which also has the structure of a Poisson manifold.  \citet{FeSt24} treat the special case in which $S_M$ is already a symplectic manifold. 
 In the setting of this paper, however, $S_M$ may not possess a symplectic form. Instead, we will show directly that $S_M$ is a Lagrangian relation determined by an equivalence class of symplectic dual pairs.  Thus for the groupoid structures in the setting of this paper, we arrive at a Lagrangian relation we can understand as the classical limit of $\underline{\mathcal{E}}_M$.

We can then define the action of $L$ on arrows in $\mathbf{LC^*}$ via the association
\begin{align}
    L: [C^*(G)\rightarrowtail\mathcal{E}_M\leftarrowtail C^*(H)]\mapsto [S_M],
\end{align}
where $[S_M]$ denotes the equivalence class of symplectic dual pairs that determine the Lagrangian relation $S_M$.
\citet{FeSt24} establish that the symplectic classical limit respects composition of arrows by the interior tensor product of Hilbert bimodules in $\mathbf{LC^*}$ and the tensor product of symplectic dual pairs in $\mathbf{LPoisson}$, up to unitary equivalence of Hilbert bimodules and isomorphism of symplectic dual pairs.  In fact, this ensures that $L$ is well-defined on our categories---to see this, suppose two Hilbert bimodules $C^*(G)\rightarrowtail\mathcal{E}_M\leftarrowtail C^*(H)$ and $C^*(G)\rightarrowtail\mathcal{E}_{M'}\leftarrowtail C^*(H)$ determine equivalent induction functors.  Then $\mathcal{E}_M$ and $\mathcal{E}_{M'}$ must have the same tensor product kernel, and it follows by direct inspection of the construction in Eq.~\eqref{eq:symplectic_classical_limit} that $L(\mathcal{E}_M)= L(\mathcal{E}_{M'})$, since they determine the same Lagrangian relation.  Thus, $L$ is well-defined, and the result from \citet{FeSt24} guarantees that $L$ is a functor.

With the functors $Q$ and $L$ so specified, the remainder of this paper is devoted to (i) showing that $L$ indeed has the properties just claimed, and (ii) showing that $Q$ and $L$ form a categorical equivalence.  The main result establishing the categorical equivalence requires showing that $Q$ and $L$ are almost inverse to each other in the sense that $Q\circ L$ and $L\circ Q$ are naturally isomorphic to the identity functors.  Ultimately, the crux of the issue is to show that whenever $Q\left([T^*M]\right) = [\mathcal{E}_M]$ and $[S_M] = L\left([\mathcal{E}_M]\right)$, it must follow that $S_M\cong (j_G,j_H)[T^*M]$.  We will prove this main result, and thus establish the existence of a categorical equivalence between $\mathbf{LPoisson}$ and $\mathbf{LC^*}$, in the remainder of the paper.

\section{Categorical Equivalence}
\label{sec:equivalence}

\subsection{Bimodules between Continuous Bundles}

Recall that we consider the quantizations $(\mathcal{Q}_\hbar^G)_{\hbar\in (0,1]}$ and $(\mathcal{Q}_\hbar^H)_{\hbar\in (0,1]}$ for Lie groupoids $G$ and $H$ as generating uniformly continuous bundles with C*-algebras of sections $\mathfrak{A}^G$ and $\mathfrak{A}^H$, respectively.  In this section, we will construct a Hilbert bimodule (denoted $\underline{\mathcal{E}}_M$) between the algebras of continuous sections $\mathfrak{A}^G$ and $\mathfrak{A}^H$ from an arrow $C^*(G)\rightarrowtail \mathcal{E}_M\leftarrowtail C^*(H)$ in $\mathbf{LC^*}$, understood as a Hilbert bimodule between the fiber C*-algebras $\mathfrak{A}^G_\hbar=C^*(G)$ and $\mathfrak{A}^H_\hbar=C^*(H)$ at a fixed value $\hbar\in (0,1]$.

To that end, we consider the space $(0,1]\times M$ for the manifold $M$ that serves as the middle space of the bibundle $G\rightarrowtail M\leftarrowtail H$.  Paralleling the construction of $\mathcal{E}_M$ after Eq. (\ref{eq:inner_product}), we begin with the vector space $C^\infty((0,1]\times M)$.   We will understand an element $\underline{\varphi}\in C^\infty((0,1]\times M)$ as corresponding to a map $(0,1]\to C^\infty(M)$ with $\hbar\mapsto \underline{\varphi}_\hbar(\cdot) = \underline{\varphi}(\hbar,\cdot)$. We define the subspace
\begin{align}
    \widehat{\underline{\mathcal{E}}}_M = \Big\{\underline{\varphi}\in C^\infty((0,1]\times M)\ \Big|\ \text{ for each } \hbar\in(0,1],\  \underline{\varphi}_\hbar\in C_c^\infty(M) \Big\}
\end{align}
We will see that a completion of a subspace of $\widehat{\underline{\mathcal{E}}}_M$ will define the middle space of a Hilbert bimodule denoted $\underline{\mathcal{E}}_M$.  

Note similarly that since continuous sections in $\mathfrak{A}^G$ are maps $(0,1]\to C^*(G)$, we will understand an element $a\in \mathfrak{A}^G$ as a map $(0,1]\times G\to \mathbb{C}$.  Indeed, viewing $\mathfrak{A}^G$ as the normal groupoid of $G$ as in \citet{La99}, one sees $\mathfrak{A}^G$ as a suitable completion of a subspace of $C^\infty((0,1]\times G)$, understood with an appropriate smooth manifold structure.  Similarly, since continuous sections in $\mathfrak{A}^H$ are maps $(0,1]\to C^*(H)$, we will understand an element $b\in \mathfrak{A}^H$ as a map $(0,1]\times H\to \mathbb{C}$, since $\mathfrak{A}^H$ is also a suitable completion of a subspace of $C^\infty((0,1]\times H)$.  For ease of notation, we will write $a_\hbar$ for $a(\hbar)$ and $b_\hbar$ for $b(\hbar)$.

Define a bilinear map on $\widehat{\underline{\mathcal{E}}}_M$ by
\begin{align}
    \inner{\underline{\varphi}}{\underline{\psi}}_{\underline{\mathcal{E}}_M}(\hbar,y) = \inner{\underline{\varphi}_\hbar}{\underline{\psi}_\hbar}_{\mathcal{E}_M}(y) = \int_{s_M^{-1}(t_H(y))}d\mu_{t_H(y)}(q)\ \overline{\underline{\varphi}(\hbar,q)}\underline{\psi}(\hbar,qy)
\end{align}
for all $\underline{\varphi},\underline{\psi}\in \widehat{\underline{\mathcal{E}}}_M$, $\hbar\in (0,1]$, and $y\in H$.  Then define the subspace $\widetilde{\underline{\mathcal{E}}}_M\subseteq\widehat{\underline{\mathcal{E}}}_M$ by
\begin{align}
    \widetilde{\underline{\mathcal{E}}}_M = \Big\{\underline{\varphi}\in \widehat{\underline{\mathcal{E}}}_M\ \Big|\ \hbar\mapsto \underline{\varphi}_\hbar\text{ is uniformly continuous with respect to }\norm{\cdot}_{\mathcal{E}_M}\Big\}. 
\end{align}

\noindent Clearly, $\widetilde{\underline{\mathcal{E}}}_M$ is a vector space.  Then $\underline{\mathcal{E}}_M$ is defined as the completion of $\widetilde{\underline{\mathcal{E}}}_M$ relative to the norm $\norm{\underline{\varphi}}_{\underline{\mathcal{E}}_M}^2 = \norm{\inner{\underline{\varphi}}{\underline{\varphi}}_{\underline{\mathcal{E}}_M}}$.

This vector space $\underline{\mathcal{E}}_M$ will serve as the middle space of a Hilbert bimodule between the algebras of continuous sections.  The left $\mathfrak{A}^G$-action and right $\mathfrak{A}^H$-action on $\underline{\mathcal{E}}_M$ are given for $a\in\mathfrak{A}^G$ and $b\in\mathfrak{A}^H$ by
\begin{align}
    (a\cdot \underline{\varphi})(\hbar,q) &= (a_\hbar\cdot \underline{\varphi}_\hbar)(q) = \int_{t_G^{-1}(t_M(q))} d\nu^t_{t_M(q)}(x)\ a(\hbar,x)\underline{\varphi}(\hbar,x^{-1}q)\\
    (\underline{\varphi}\cdot b)(\hbar,q) &= (\underline{\varphi}_\hbar\cdot b_\hbar)(q) = \int_{t_H^{-1}(s_M(q))} d\lambda^t_{s_M(q)}(y)\ b(\hbar,y^{-1})\underline{\varphi}(\hbar,qy)
\end{align}
for all $\underline{\varphi}\in \widehat{\underline{\mathcal{E}}}_M$, $q\in M$, and $\hbar\in (0,1]$, extended to $\underline{\mathcal{E}}_M$ by continuity.

We now proceed to show that $\mathfrak{A}^G\rightarrowtail \underline{\mathcal{E}}_M\leftarrowtail \mathfrak{A}^H$ is indeed a Hilbert bimodule.  The only feature that must be checked is that $\inner{\underline{\varphi}}{\underline{\psi}}_{\underline{\mathcal{E}}_M}\in \mathfrak{A}^H$ for all $\underline{\varphi},\underline{\psi}\in \underline{\mathcal{E}}_M$ to ensure that $\inner{\cdot}{\cdot}_{\underline{\mathcal{E}}_M}$ is an $\mathfrak{A}^H$-valued inner product.  We will need a preliminary lemma concerning the inner product $\inner{\cdot}{\cdot}_{\mathcal{E}_M}$.\bigskip

\begin{lemma}
\label{lemma:continuity_inner_product}
For any $\underline{\varphi},\underline{\psi}\in\widetilde{\underline{\mathcal{E}}}_M$, the map $\hbar\mapsto \inner{\underline{\varphi}_\hbar}{\underline{\psi}_\hbar}_{\mathcal{E}_M}$ is uniformly continuous with respect to the norm on $C^*(H)$.
\end{lemma}

\begin{proof}
    Suppose $\epsilon>0$.  Since $\hbar\mapsto \underline{\varphi}_\hbar$ and $\hbar\mapsto\underline{\psi}_\hbar$ are uniformly continuous, there are $\delta_\varphi>0$ and $\delta_\psi>0$ such that for any $\hbar,\hbar'\in (0,1]$ with $\abs{\hbar-\hbar'}<\epsilon,$
\begin{align}
    &\norm{\underline{\varphi}_\hbar-\underline{\varphi}_{\hbar'}}_{\mathcal{E}_M}<\delta_\varphi\\
    &\norm{\underline{\psi}_\hbar-\underline{\psi}_{\hbar'}}_{\mathcal{E}_M}<\delta_\psi.
\end{align}
We define
    \begin{align}
        \delta = \delta_\psi\cdot \sup_{\hbar\in (0,1]}\norm{\underline{\varphi}_\hbar}_{\mathcal{E}_M} + \delta_\varphi\cdot \sup_{\hbar\in (0,1]}\norm{\underline{\psi}_\hbar}_{\mathcal{E}_M}.
    \end{align}
    Then for any $\hbar,\hbar'\in (0,1]$ with $\abs{\hbar-\hbar'}<\epsilon$, the Cauchy-Schwarz inequality for Hilbert bimodules \citep[Lemma 2.5]{RaWi98} implies
    \begin{align}
        \norm{\inner{\underline{\varphi}_\hbar}{\underline{\psi}_\hbar}_{\mathcal{E}_M} - \inner{\underline{\varphi}_{\hbar'}}{\underline{\psi}_{\hbar'}}_{\mathcal{E}_M}} &= \norm{\inner{\underline{\varphi}_\hbar}{\underline{\psi}_\hbar}_{\mathcal{E}_M} - \inner{\underline{\varphi}_\hbar}{\underline{\psi}_{\hbar'}}_{\mathcal{E}_M} + \inner{\underline{\varphi}_\hbar}{\underline{\psi}_{\hbar'}}_{\mathcal{E}_M} -\inner{\underline{\varphi}_{\hbar'}}{\underline{\psi}_{\hbar'}}_{\mathcal{E}_M}}\\
        &\leq \norm{\inner{\underline{\varphi}_\hbar}{\underline{\psi}_\hbar-\underline{\psi}_{\hbar'}}_{\mathcal{E}_M}} + \norm{\inner{\underline{\varphi}_\hbar-\underline{\varphi}_{\hbar'}}{\underline{\psi}_\hbar}_{\mathcal{E}_M}}\\
        &\leq \norm{\underline{\varphi}_\hbar}_{\mathcal{E}_M}\cdot \norm{\underline{\psi}_\hbar-\underline{\psi}_{\hbar'}}_{\mathcal{E}_M} + \norm{\underline{\varphi}_\hbar - \underline{\varphi}_{\hbar'}}_{\mathcal{E}_M}\cdot \norm{\underline{\psi}_{\hbar'}}_{\mathcal{E}_M}\\
        &<\delta.
    \end{align}
\end{proof}

\noindent Next, to compare the inner product $\inner{\cdot }{\cdot}_{\underline{\mathcal{E}}_M}$ with the continuous sections in $\mathfrak{A}^H$, we shall need some further results concerning the generating set $\tilde{\mathfrak{A}}^H$ defined (see discussion around Eq.\eqref{eq:quantization_sections}) as the *-algebra of pointwise products and sums of sections of the form
\begin{align}
    [\hbar\mapsto \mathcal{Q}_\hbar^H(g)]
\end{align}
for $g\in C^\infty_{PW}(\mathfrak{H}^*)$.\bigskip

\begin{lemma}
    Suppose $g\in C^\infty_{PW}(\mathfrak{H}^*)$.  The map $\hbar\mapsto \pi^H(\mathcal{Q}_\hbar(g))$ is norm continuous for $\hbar\in (0,1]$, i.e.,
    \begin{align}
        \lim_{\hbar\to\hbar'}\norm{\pi^H(\mathcal{Q}_\hbar^G(g)-\pi^H(\mathcal{Q}_{\hbar'}^H(g))} = 0.
    \end{align}
\end{lemma}

\begin{proof}
    The map $\hbar\mapsto K^{\mathcal{Q}_\hbar^H(g)}$ yielding the Hilbert-Schmidt kernel (Eq.\eqref{eq:HSkernel}) clearly yields a pointwise continuous family of functions on $H\times H$.  The dominated convergence theorem in $L^2(H\times H)$ implies that
    \begin{align}
        \lim_{\hbar\to\hbar'}\norm{\pi^H(\mathcal{Q}_\hbar(g)) - \pi^H(\mathcal{Q}_{\hbar'}^G(g))}_{HS} = 0,
    \end{align}
    where $\norm{F}_{HS} = \norm{K_F}_{L^2(H\times H)}$ denotes the Hilbert-Schmidt norm for the kernel $K_F$ of any Hilbert-Schmidt operator $F$.  Now the fact that $\norm{F}\leq \norm{F}_{HS}$ yields the result.
\end{proof}

\begin{lemma}
\label{lemma:continuity_generating_sections}
    Suppose $g_1,...,g_k\in C^\infty_{PW}(\mathfrak{H}^*)$.  The map $\hbar\mapsto \pi^H(\mathcal{Q}_\hbar^H(g_1)\cdot...\cdot \mathcal{Q}_\hbar^G(g_k))$ is norm continuous.
\end{lemma}

\begin{proof}
    By induction.  The previous proposition yields the base case.  Suppose the condition holds for $g_1,...,g_{k-1}$.  Then using the triangle inequality, one easily finds
    \begin{align}
        \lim_{\hbar\to\hbar'}&\norm{\pi^H(\mathcal{Q}_\hbar^H(g_1)\cdot...\cdot\mathcal{Q}_\hbar^H(g_k)) - \pi^H(\mathcal{Q}_{\hbar'}^H(g_1)\cdot...\cdot \mathcal{Q}_{\hbar'}^H(g_k))}\\
        &\leq\lim_{\hbar\to\hbar'}\norm{\pi^H(\mathcal{Q}_\hbar^H(g_1)\cdot...\cdot\mathcal{Q}_\hbar^H(g_{k-1})) - \pi^H(\mathcal{Q}_{\hbar'}^H(g_1)\cdot...\cdot\mathcal{Q}_{\hbar'}^H(g_{k-1}))}\cdot\norm{\pi^H(\mathcal{Q}_{\hbar'}(g_k))}\\
       & \hspace{.75in} + \norm{\pi^H(\mathcal{Q}_\hbar^H(g_1)\cdot...\cdot\mathcal{Q}_\hbar^H(g_{k-1}))}\cdot \norm{\pi^H(\mathcal{Q}_\hbar^G(g_k)) - \pi^H(\mathcal{Q}_{\hbar'}^G(G_k))}\\
        &=0.
    \end{align}
\end{proof}

The analogous statements clearly hold for $G$ just as well as for $H$.  We are now in a position use these preliminary lemmas ot establish that $\underline{\mathcal{E}}_M$ is the middle space of a Hilbert bimodule.\bigskip

\begin{proposition}
      For any $\underline{\varphi},\underline{\psi}\in\widetilde{\underline{\mathcal{E}}}_M$, $\inner{\underline{\varphi}}{\underline{\psi}}_{\underline{\mathcal{E}}_M}\in\mathfrak{A}^H.$ 
\end{proposition}

\begin{proof}
Recall from Eq. (\ref{eq:generating_sections}) that $b\in \mathfrak{A}^H$ just in case for each $\tilde{b}\in \tilde{\mathfrak{A}}^H$, the map $\hbar\mapsto\norm{b_\hbar - \tilde{b}_\hbar}$ is continuous for all $\hbar\in [0,1]$. Alternatively, this is equivalent to saying $\hbar\mapsto\norm{b_\hbar - \tilde{b}_\hbar}$ is uniformly continuous for all $\hbar\in (0,1]$.  Lemma \ref{lemma:continuity_generating_sections} implies that for each $\tilde{b}\in\tilde{\mathfrak{A}}^H$ and any $\epsilon>0$, there is a $\delta_b$ such that if $\abs{\hbar-\hbar'}<\epsilon$, then
\begin{align}
    \norm{\tilde{b}_\hbar - \tilde{b}_{\hbar'}}<\delta_b.
\end{align}
Likewise, Lemma \ref{lemma:continuity_inner_product} implies that there is a $\delta$ such that if $\abs{\hbar-\hbar'}<\epsilon$, then
\begin{align}
    \norm{\inner{\underline{\varphi}_\hbar}{\underline{\psi}_\hbar}_{\mathcal{E}_M} - \inner{\underline{\varphi}_{\hbar'}}{\underline{\psi}_{\hbar'}}_{\mathcal{E}_M}}<\delta
\end{align}
So we have from the reverse triangle inequality that whenever $\abs{\hbar-\hbar'}<\epsilon$
\begin{align}
\Big|\norm{\inner{\underline{\varphi}_\hbar}{\underline{\psi}_\hbar}_{\mathcal{E}_M} - \tilde{b}_\hbar} - \norm{\inner{\underline{\varphi}_{\hbar'}}{\underline{\psi}_{\hbar'}}_{\mathcal{E}_M} - \tilde{b}_{\hbar'}}\Big|&\leq \norm{\inner{\underline{\varphi}_\hbar}{\underline{\psi}_\hbar}_{\mathcal{E}_M} - \inner{\underline{\varphi}_{\hbar'}}{\underline{\psi}_{\hbar'}}_{\mathcal{E}_M}} + \norm{\tilde{b}_\hbar - \tilde{b}_{\hbar'}}\\
&<\delta + \delta_b.
\end{align}
This establishes the uniform continuity of the map $\hbar\mapsto \norm{\inner{\underline{\varphi}_\hbar}{\underline{\psi}_\hbar}_{\mathcal{E}_M}-\tilde{b}_\hbar}$ for any $\tilde{b}\in\tilde{\mathfrak{A}}^H$, which by the definition in Eq.\eqref{eq:generating_sections} implies that $\inner{\underline{\varphi}}{\underline{\psi}}_{\underline{\mathcal{E}}_M}\in\mathfrak{A}^H$.
\end{proof}

\begin{cor}
\label{cor:lift}
$\mathfrak{A}^G\rightarrowtail\underline{\mathcal{E}}_M\leftarrowtail\mathfrak{A}^H$ is a Hilbert bimodule.
\end{cor}\bigskip
\noindent We call the Hilbert bimodule $\underline{\mathcal{E}}_{M}$ the \emph{lift} of $\mathcal{E}_M$.

To close this section, we demonstrate that the construction of $\underline{{\mathcal{E}}}_M$ from $\mathcal{E}_M$ respects the composition of Hilbert bimodules by means of the tensor product operation in $\mathbf{LC^*}$.  Suppose \begin{align}
    C^*(G)\rightarrowtail\mathcal{E}_M\leftarrowtail C^*(H) && \text{and} && C^*(H)\rightarrowtail\mathcal{E}_N\leftarrowtail C^*(G').
    \end{align}
    are two Hilbert bimodules in $\mathbf{LC^*}$ between Lie groupoids $G, H,$ and $G'$ corresponding to bibundles
    \begin{align}
        G\rightarrowtail M\leftarrowtail H && \text{and} && H\rightarrowtail N\leftarrowtail G'.
    \end{align}
\citet{La01a} already shows that the tensor product $\mathcal{E}_M\otimes_{C^*(H)}\mathcal{E}_N$ is unitarily equivalent to $\mathcal{E}_{M\circledast_H N}$. 
 Hence, one can use the procedure leading to Corollary \ref{cor:lift} to define the lift of $\mathcal{E}_M\otimes_{C^*(H)}\mathcal{E}_N$ as the lift of $\mathcal{E}_{M\circledast_H N}$.\bigskip

\begin{proposition}
    
The Hilbert bimodule $\underline{\mathcal{E}}_M\otimes_{\mathfrak{A}^H}\underline{\mathcal{E}}_N$ is unitarily equivalent to the lift of $\mathcal{E}_M\otimes_{C^*(H)}\mathcal{E}_N$.
\end{proposition}

\begin{proof}
    We need to show that $ \underline{\mathcal{E}}_M\otimes_{\mathfrak{A}^H}\underline{\mathcal{E}}_N\cong \underline{\mathcal{E}}_{M\circledast_H N}$, where the latter is by definition the lift of $\mathcal{E}_{M\circledast_H N}$.  \citet[][Eq. 3.8, p. 107]{La01a} explicitly defines a unitary equivalence $U: \mathcal{E}_M\otimes_{C^*(H)}\mathcal{E}_N\to \mathcal{E}_{M\circledast_H N}$.  Now, define $\underline{U}: \underline{\mathcal{E}}_M\otimes_{\mathfrak{A}^H}\underline{\mathcal{E}}_N\to \underline{\mathcal{E}}_{M\circledast_H N}$ by
    \begin{align}
        (\underline{U}\underline{\varphi})_\hbar = U(\underline{\varphi}_\hbar)
    \end{align}
    for $\underline{\varphi}\in \underline{\mathcal{E}}_M\otimes_{\mathfrak{A}^H}\underline{\mathcal{E}}_N$ and $\hbar\in (0,1]$.  It is easy to see that $\underline{U}$ has the appropriate range and serves also as a unitary equivalence.
\end{proof}

\noindent Finally, we note that one can check that the association $\mathcal{E}_M\mapsto \underline{\mathcal{E}}_M$ preserves identities, and so is functorial.

In order to take the classical limit of the Hilbert bimodule $\underline{\mathcal{E}}_M$, we must check that it is strongly nondegenerate at $\hbar =0$ in the sense of \citet{FeSt24}.\bigskip

\begin{proposition}
    Define
    \begin{align}
    K_0^G = \left\{a\in\mathfrak{A}^G\ |\ \lim_{\hbar\to 0} ||a_\hbar|| = 0 \right\} && K_0^H = \left\{b\in\mathfrak{A}^H\ | \lim_{\hbar\to 0}||b_\hbar|| = 0\right\}.
    \end{align}
    Then $\overline{K_0^G \cdot \underline{\mathcal{E}}_M}\subseteq \overline{\underline{\mathcal{E}}_M\cdot K_0^H}$.
\end{proposition}

\begin{proof}
    Let $\underline{\psi}\in\overline{K_0^G \cdot \underline{\mathcal{E}}_M}$.  Then there are $\overset{n}{a}\in K_0^G$ with $\lim_{\hbar\to 0}||\overset{n}{a}_\hbar|| = 0$ and $\overset{n}{\underline{\varphi}}\in\underline{\mathcal{E}}_M$ such that $\underline{\psi} = \sum_n \overset{n}{a}\cdot\overset{n}{\underline{\varphi}}$.  Note that $\mathfrak{A}^H$ has a sequential approximate identity $\overset{m}{b}\in\mathfrak{A}^H$: e.g., one can take a sequence of functions $\overset{m}{b}(\hbar,y)\in C^\infty((0,1]\times H)$ such that for each $\hbar\in (0,1]$, $\overset{m}{b}(\hbar,y)$ approaches a delta function around $y=e$ in $H$ as $m\to \infty$ (with the $m\to \infty$ limit understood in any appropriate distributional sense).  Then it follows that for any $\underline{\varphi}\in\underline{\mathcal{E}}_M$,
    \begin{align}
        ||\underline{\varphi} - \underline{\varphi}\cdot\overset{m}{b}||_{\underline{\mathcal{E}_M}}\to 0
    \end{align}
    as $m\to \infty$.  Hence, defining $\overset{n,m}{b'}_\hbar = ||\overset{n}{a}_\hbar||\cdot\overset{m}{b}_\hbar$, we have that $\overset{n,m}{b'}\in\mathfrak{A}^H$ by the definition of a uniformly continuous bundle, and clearly $\lim_{\hbar\to0}||\overset{n,m}{b'}_\hbar||= 0$ so it follows that $\overset{n,m}{b'}\in K_0^H$.  Likewise, defining $\overset{n}{\underline{\varphi}'}_\hbar = \left(||\overset{n}{a}_\hbar||^{-1}\cdot\overset{n}{a}_\hbar\right)\cdot \overset{n}{\underline{\varphi}}$, we have $\overset{n}{\underline{\varphi}'}\in\underline{\mathcal{E}}_M$ and
    \begin{align}
        \underline{\psi} = \lim_{m\to\infty}\sum_n \overset{n}{\underline{\varphi}'} \cdot \overset{n,m}{b'}.
    \end{align}
    Therefore, $\underline{\psi}\in \overline{\underline{\mathcal{E}}_M\cdot K_0^H}$
\end{proof}
\noindent The previous argument actually establishes strong nondegeneracy under general conditions, as long as an appropriate sequential approximate identity exists for the right action.

To summarize, in this section we have shown that each Hilbert bimodule $\mathcal{E}_M$ in $\mathbf{LC^*}$ canonically and functorially lifts to a Hilbert bimodule $\underline{\mathcal{E}}_M$ between the C*-algebras of uniformly continuous sections of the corresponding bundles.  Further, the Hilbert bimodule $\underline{\mathcal{E}}_M$ is strongly nondegenerate, so the construction of \citet{FeSt24} implies that $\underline{\mathcal{E}}_M$ has a Hilbert classical limit, which we denote $(\mathcal{E}_M)_0$.  The Hilbert bimodule $C_0(\mathfrak{G}^*)\rightarrowtail (\mathcal{E}_M)_0\leftarrowtail C_0(\mathfrak{H}^*)$ at $\hbar = 0$ gives rise to canonical representations $\pi^G_0$ and $\pi^H_0$ on $(\mathcal{E}_M)_0$ of $C_0(\mathfrak{G}^*)$ and (the opposite of) $C_0(\mathfrak{H}^*)$, respectively. We can then follow \citet{FeSt24} in defining the space $S_M$ as the Gelfand dual of the quotient $(C_0(\mathfrak{G}^*)\otimes C_0(\mathfrak{H}^*))/\ker\left(\pi^G_0\otimes \pi^H_0\right)$, as in Eq.\eqref{eq:symplectic_classical_limit}.  In our setting, $S_M$ is not in general a symplectic space, but it will turn out that it is isomorphic to a Lagrangian relation $S_M\subseteq \mathfrak{G}^*\times\mathfrak{H}^*$.

\subsection{Local Structure of Classical Dual Pairs}

Recall that our goal is to show that the functors $Q$ and $L$ form a categorical equivalence, which involves establishing that $S_M$ is isomorphic to the Lagrangian relation $(j_G,j_H)[T^*M]$.  To prove this, we will require some additional results concerning the local structure of the dual pair
\begin{align}
    \mathfrak{G}^*\leftarrow T^*M\rightarrow \mathfrak{H}^*.
\end{align}
Throughout this section, we assume as before that $T^*M$ is a dual pair in $\mathbf{LPoisson}$, which arises from a regular bibundle $M$.  Our goal in this section is to establish the requisite facts about this dual pair to proceed to our categorical equivalence.

For ease of reference, we begin by recalling the definitions of the maps $j_G$ and $j_H$.  The Poisson morphism $j_G: T^*M\to \mathfrak{G}^*$ is defined by
\begin{align}
    j_G(\eta_q)\Big(\frac{d}{dt}_{|t=0} \gamma(t)\Big) = -\eta_q\Big(\frac{d}{dt}_{|t=0}\big(\gamma(t)^{-1}\cdot q\big)\Big),
\end{align}
where $\eta_q\in T^*_qM$ and $\gamma:I\to G$ is a curve with $\gamma(0) = t_M(q)$ and $\gamma(t)\in t_G^{-1}(t_M(q))$ so that $\dot{\gamma}(0)\in \tau_G^{-1}(t_M(q))$ and $\gamma(t)^{-1}\cdot q$ is well-defined for all $t\in I\subseteq \mathbb{R}$.  Likewise, the Poisson morphism $j_H: T^*M\to \mathfrak{H}^*$ is defined by
\begin{align}
    j_H(\eta_q)\Big(\frac{d}{dt}_{|t=0} \gamma(t)\Big) = \eta_q\Big(\frac{d}{dt}_{|t=0} \big(q\cdot \gamma(t)\big)\Big),
\end{align}
where $\eta_q\in T^*_qM$ and $\gamma: I\to H$ is a curve with $\gamma(0) = s_M(q)$ and $\gamma(t)\in t_H^{-1}(s_M(q))$, so that $\dot{\gamma}(0)\in \tau_H^{-1}(s_M(q))$ and $q\cdot \gamma(t)$ is well-defined for all $t\in I\subseteq \mathbb{R}$.

In what follows, we denote $\mathfrak{G}_{x_0} = \{X\in\mathfrak{G}\ |\ t_G(\exp^W(X)) = x_0\}$ and likewise, $\mathfrak{H}_{y_0} = \{Y\in\mathfrak{H}\ |\ t_H(\exp^W(Y)) = y_0\}$.  Note further that for each $q\in M$, we have a map $(j_G)^*_{|q}: \mathfrak{G}_{t_M(q)}\to T_qM$ defined by
\begin{align}
    (j_G)^*_{|q}(X) = -\frac{d}{dt}_{|t=0}\exp^W(tX)^{-1}\cdot q
\end{align}
for all $X\in \mathfrak{G}_{t_M(q)}$.  Likewise, we have a map $(j_H)^*_{|q}: \mathfrak{H}_{s_M(q)}\to T_q M$ defined by
\begin{align}
    (j_H)^*_{|q}(Y) = \frac{d}{dt}_{|t=0} q\cdot \exp^W(tY)
\end{align}
for all $Y\in \mathfrak{H}_{s_M(q)}$. 
 Using these maps, we define $j_{|q}: T^*_qM\to \mathfrak{G}^*_{t_M(q)}\oplus\mathfrak{H}^*_{s_M(q)}$ by
\begin{align}
    j_{|q} = (j_G)_{|q}\oplus (j_H)_{|q},
\end{align}
and we define $j^*_{|q}: \mathfrak{G}_{t_M(q)}\oplus\mathfrak{H}_{s_M(q)}\to T_qM$ by
\begin{align}
    j^*_{|q} = (j_G)^*_{|q} \oplus (j_H)^*_{|q}.
\end{align}
Together, our definitions read
\begin{align}
    j_{|q}(\eta_q)(Z) = (\eta_q)(j^*_{|q}(Z))
\end{align}
for all $\eta_q\in T^*_qM$ and $Z\in\mathfrak{G}_{t_M(q)}\oplus\mathfrak{H}_{s_M(q)}$.

The following propositions follow from simple dimension counting arguments. 
 We state them explicitly here only in order to establish the notation we will use in the next section.\bigskip

\begin{lemma}
    There is a linear map $\rho_q: T_qM\to \mathfrak{G}_{t_M(q)}\oplus\mathfrak{H}_{s_M(q)}$ such that \begin{align}
        j^*_{|q}\circ\rho_q(\xi_q) = \xi_q
    \end{align}
    for all $\xi_q\in T_qM$.
\end{lemma}

\begin{proof}
    Choose a basis $\overset{1}{\xi}_q,...,\overset{n}{\xi}_q$ for $T_qM$ (here, $n = \dim M$).  For each $1\leq k\leq n$, choose $Z_k\in (j^*)_{|q}^{-1}\Big(\overset{k}{\xi}_q\Big)$.  Then for $\xi_q = \sum_k \alpha_k\cdot \overset{k}{\xi}_q$, define
    \begin{align}
        \rho_q(\xi_q) = \sum_k \alpha_k\cdot Z_k.
    \end{align}
    Clearly, $\rho_q$ is linear by construction and
    \begin{align}
        j^*_{|q}\circ\rho_q(\xi_q) = \sum_k\alpha_k\cdot j^*_{|q}(Z_k) = \sum_k\alpha_k\cdot \overset{k}{\xi}_q = \xi_q.
    \end{align}
\end{proof}

Now, fix a choice of linear map $\rho_q$ satisfying the conditions of the previous lemma.  Note $T_qM = \text{im}(j^*_{|q})$ and define $\mathcal{I}_q = \ker(j^*_{|q})$.  Then define a further mapping denoted $\iota_q: \mathfrak{G}_{t_M(q)}\oplus\mathfrak{H}_{s_M(q)}\to T_qM\oplus\mathcal{I}_q$ by
\begin{align}
    \iota_q(Z) = j^*_{|q}(Z)\oplus \big(Z-\rho_q\circ j^*_{|q}(Z)\big)
\end{align}
for all $Z\in \mathfrak{G}_{t_M(q)}\oplus\mathfrak{H}_{s_M(q)}.$\medskip

\begin{proposition}
    The map $\iota_q$ is an isomorphism $\mathfrak{G}_{t_M(q)}\oplus\mathfrak{H}_{s_M(q)}\cong T_qM\oplus\mathcal{I}_q$.
\end{proposition}
\begin{proof}
    Define the map $\iota_q^{-1}: T_qM\oplus \mathcal{I}_q\to \mathfrak{G}_{t_M(q)}\oplus\mathfrak{H}_{s_M(q)}$ by
    \begin{align}
        \iota_q^{-1}(\xi_q\oplus Z) = \rho_q(\xi_q) + Z
    \end{align}
    for $\xi_q\in T_qM$ and $Z\in \mathcal{I}_q$.  Then $\iota_q^{-1}$ is indeed the inverse of $\iota_q$ because for all $Z\in \mathfrak{G}_{t_M(q)}\oplus\mathfrak{H}_{s_M(q)}$, we have
    \begin{align}
        \iota_q^{-1}\circ \iota_q(Z) &= \rho_q\circ j^*_{|q}(Z) + Z - \rho_q\circ j^*_{|q}(Z)\\
        &=Z.
    \end{align}
    And for all $\xi_q\in T_qM$ and $Z\in\mathcal{I}_q$, we have
    \begin{align}
        \iota_q\circ\iota_q^{-1}(\xi_q\oplus Z) &= j^*_{|q}(\rho_q(\xi_q) + Z)\oplus \big(\rho_q(\xi_q)+Z - \rho_q\circ j^*_{|q}(\rho_q(\xi_q)+Z)\big)\\
        &=\xi_q\oplus Z.
    \end{align}
\end{proof}

Likewise, define the dual map $\iota^*_q: T^*_qM\oplus \mathcal{I}_q^*\to \mathfrak{G}^*_{t_M(q)}\oplus\mathfrak{H}^*_{s_M(q)}$ by
\begin{align}
    \iota^*_q(\eta_q\oplus\Phi)(Z) = (\eta_q\oplus\Phi)(Z))
\end{align}
for all $\eta_q\in T^*_qM$, $\Phi\in \mathcal{I}_q^*$, and $Z\in \mathfrak{G}_{t_M(q)}\oplus\mathfrak{H}_{s_M(q)}.$\bigskip

\begin{proposition}
    The map $\iota^*_q$ is an isomorphism $T^*_qM\oplus \mathcal{I}_q^*\cong\mathfrak{G}^*_{t_M(q)}\oplus\mathfrak{H}^*_{s_M(q)}$.
\end{proposition}

\begin{proof}
    Define the map $(\iota^*_q)^{-1}: \mathfrak{G}_{t_M(q)}^*\oplus \mathfrak{H}_{s_M(q)}^*\to T^*_qM\oplus\mathcal{I}_q^*$ by
    \begin{align}
        (\iota^*_q)^{-1}(\Phi)(\xi_q\oplus Z) = \Phi(\iota_q^{-1}(\xi_q\oplus Z))
    \end{align}
for $\Phi\in \mathfrak{G}_{t_M(q)}^*\oplus \mathfrak{H}_{s_M(q)}^*$, $\xi_q\in T_qM$, and $Z\in \mathcal{I}_q$.  Then $(\iota_q^*)^{-1}$ is indeed the inverse of $\iota_q^*$ because for all $\eta_q\in T_q^*M$, $\Phi\in \mathcal{I}_q^*$, $\xi_q\in T_qM$, and $Z\in \mathcal{I}_q$, we have
\begin{align}
    \big((\iota_q^*)^{-1}\circ\iota_q^*\big)(\eta_q\oplus\Phi)(\xi_q\oplus Z) &= \iota_q^*(\eta_q\oplus \Phi)(\iota_q^{-1}(\xi_q\oplus Z))\\
    &= (\eta_q\oplus \Phi)(\iota_q\circ\iota_q^{-1}(\xi_q\oplus Z))\\
    &=(\eta_q\oplus \Phi)(\xi_q\oplus Z),
\end{align}
and for all $\Phi\in \mathfrak{G}_{t_M(q)}^*\oplus \mathfrak{H}_{s_M(q)}^*$ and $Z\in \mathfrak{G}_{t_M(q)}\oplus\mathfrak{H}_{s_M(q)}$, we have
\begin{align}
    \big(\iota^*_q\circ(\iota_q^*)^{-1}\big)(\Phi)(Z) &= (\iota_q^*)^{-1}(\Phi)(\iota_q(Z))\\
    &= \Phi(\iota_q^{-1}\circ\iota_q(Z))\\
    &= \Phi(Z).
\end{align}

\end{proof}

\begin{proposition}
    For any $\eta_q\in T^*_qM$, we have
    \begin{align}
        \iota_q^*(\eta_q\oplus 0) = j_{|q}(\eta_q).
    \end{align}
\end{proposition}

\begin{proof}
    For any $\eta_q\in T_q^*M$ and any $Z\in\mathfrak{G}_{t_M(q)}\oplus\mathfrak{H}_{s_M(q)}$, we have
\begin{align}
    \iota_q^*(\eta_q\oplus 0)(Z) &= (\eta_q\oplus 0)(\iota_q(Z))\\
    &= \eta_q(j^*_{|q}(Z))\\
    &= j_{|q}(\eta_q)(Z).
\end{align}
\end{proof}\bigskip

\noindent We will proceed to use these preliminaries in the next section to analyze the relationship between the spaces $T^*M$ and $S_M$.

\subsection{Natural Isomorphism for the Classical Limit}

In this section, we prove the main results of this paper.  We will show that $S_M$ is isomorphic to the Lagrangian relation determined by $T^*M$.  This guarantees that $L$ is well-defined, with $S_M$ the canonical invariant of the equivalence class of symplectic dual pairs $L\circ Q\left([T^*M]\right)$. More substantively, the result implies the existence of natural isomorphisms $L\circ Q\to 1_{\mathbf{LC^*}}$ and $Q\circ L\to 1_{\mathbf{LPoisson}}$.

We begin with some preliminaries.  For $q\in M$, let $E_q$ denote the Jacobian such that
\begin{align}
    \int_{t_G^{-1(t_M(q))}}&\int_{t_H^{-1}(s_M(q))} d\nu^t_{t_M(q)}(x)\ d\lambda^t_{s_M(q)(y)}\ F(x,y)\\
    &= \int_{\mathfrak{G}_{t_M(q)}}\int_{\mathfrak{H}_{s_M(q)}} E_q(X,Y)\ d^nX\ d^mY\ F(\exp^W(X),\exp^W(Y))
\end{align}
for $F: G\times H\to \mathbb{C}$ integrable on the fibers of $t_G\times t_H$.  Here, $d^nX$ and $d^mY$ are the Lebesgue measures on $\mathfrak{G}_{t_M(q)}$ and $\mathfrak{H}_{s_M(q)}$ with $n = \dim G$ and $m = \dim H$.

Given two functions $f\in \mathfrak{G}^*\to\mathbb{C}$ and $g\in \mathfrak{H}^*\to\mathbb{C}$, we will use the shorthand $(f\cdot g):\mathfrak{G}^*\oplus \mathfrak{H}^*\to\mathbb{C}$ to denote the function
\begin{align}
    (f\cdot g)(\theta\oplus \phi) = f(\theta)\cdot g(\phi)
\end{align}
for $\theta\in\mathfrak{G}^*$ and $\phi\in\mathfrak{H}^*$.

To prove that $S_M\cong T^*M$, we shall show that $C_0(S_M)\cong C_0(T^*M)$.  Recall that, by definition in Eq. (\ref{eq:symplectic_classical_limit}), we have $C_0(S_M) = (\mathfrak{A}^G_0\otimes\mathfrak{A}^H_0)/\ker(\pi^G_0\otimes\pi_0^H)$.  Our first task is to analyze $\ker(\pi_0^G\otimes \pi_0^H)$.  To do so, we likewise make use of the representations at $\hbar=1$ denoted $\pi_1^G$ and $\pi^H_1$ of $C^*(G)$ and (the opposite of) $C^*(H)$ on $\mathcal{E}_M$ given by the left $C^*(G)$-action and right $C^*(H)$-action of the Hilbert bimodule.  We will first characterize $\ker(\pi_1^G\otimes \pi_1^H)$ as a step towards characterizing $\ker(\pi_0^G\otimes \pi_0^H)$.

Given $\xi_q\in T_qM$, define $\beta_\hbar^{\xi_q}: C_c^\infty(\mathfrak{G}^*)\otimes C_c^\infty(\mathfrak{H}^*)\to C_c^\infty(T^*M)$ by
\begin{align}
    \beta_\hbar^{\xi_q}&(f\otimes g)(\eta_q)\\
    &= \int_{\mathcal{I}_q}\int_{\mathcal{I}_q^*} \frac{E_q\circ \iota_q^{-1}(\xi_q\oplus Z)\ d^{m}Z d^{m}\Phi}{\hbar^m} \big((\kappa_G\cdot\kappa_H)\circ\iota_q^{-1}\big)(\xi_q\oplus Z)\\
    & \hspace{2.5in} \times \big((f\cdot g)\circ\iota_q^*\big)(\eta_q\oplus \Phi)e^{-i\Phi(Z)/\hbar}
\end{align}
for all $f\in C_c^\infty(\mathfrak{G}^*)$, $g\in C_c^\infty(\mathfrak{H}^*)$, and $\eta_q\in T^*_qM$.\bigskip

\begin{proposition}
\label{prop:bimodule_limit}
    For all $\hbar\in (0,1]$, $f\in C_c^\infty(\mathfrak{G}^*)$, and $g\in C_c^\infty(\mathfrak{H}^*)$, we have
    \begin{align}
    \big(\mathcal{Q}_\hbar^G(f)\otimes \mathcal{Q}_\hbar^H(g)\big)\in \ker(\pi_1^G\otimes \pi_1^H)
    \end{align}
    iff for almost all $q\in M$, $\xi_q\in T_qM$ and $\eta_q\in T_q^*M$,
\begin{align}
    \beta_\hbar^{\xi_q}(f\otimes g)(\eta_q) = 0.
\end{align}
\end{proposition}\bigskip

\begin{proof}
    For $\psi\in\mathcal{E}_M$, we have
    \begin{align}
        &\Big(\big(\pi_1(\mathcal{Q}_\hbar^G(f))\otimes\pi_1^H(\mathcal{Q}_\hbar^H(g))\big)\psi\Big)(q) = \big(\mathcal{Q}_\hbar^G(f)\cdot \psi\cdot \mathcal{Q}_\hbar^H(g)\big)(q)\\
        &= \int_{t_G^{-1}(t_M(q))}\int_{t_H^{-1}(s_M(q))} d\nu^t_{t_M(q)}(x)\ d\lambda^t_{s_M(q)}(y)\ \big(\mathcal{Q}_\hbar^G(f)\big)(x)\psi(x^{-1}qy)\big(\mathcal{Q}_\hbar^H(g)\big)(y^{-1})\\
        &= \int_{\mathfrak{G}_{t_M(q)}}\int_{\mathfrak{H}_{s_M(q)}} \frac{E_q(X,Y)\ d^nX\ d^mY}{\hbar^{(n+m)}}\kappa_G(X)\kappa_H(Y)\\
        &\hspace{1in} \times \hat{f}(X/\hbar)\psi\big(\exp^W(X)^{-1}\cdot q\cdot \exp^W(Y)^{-1}\big)\hat{g}(Y/\hbar)\\
        &= \int_{\mathfrak{G}_{t_M(q)}\oplus \mathfrak{H}_{s_M(q)}}\int_{\mathfrak{G}^*_{t_M(q)}\oplus \mathfrak{H}^*_{s_M(q)}}\frac{E_q(X,Y)\ d^nX\ d^mY\ d^n\theta\ d^m\phi}{\hbar^{(n+m)}}\kappa_G(X)\kappa_H(Y)\\
        &\hspace{.75in}\times f(\theta)g(\phi)e^{-i\big(\theta(X)+\phi(Y)\big)/\hbar}\psi\big(\exp^W(X)^{-1}\cdot q\cdot \exp^W(Y)^{-1}\big).
     \end{align}
     It follows that $\pi_1^G(\mathcal{Q}_\hbar(f))\otimes\pi_1^H(\mathcal{Q}_\hbar(g))$ is an integral operator with kernel $K^{f\otimes g}: M\times M\to \mathbb{C}$ given by
     \begin{align}
        &K^{f\otimes g}\Big(q,\exp^W\big(\iota_q^{-1}(\xi_q\oplus 0)\big)\cdot q\Big)\\ &= \int_{T^*_qM}\int_{\mathcal{I}_q}\int_{\mathcal{I}_q^*} \frac{(E_q\circ\iota_q^{-1})(\xi_q\oplus Z)\ d^n\eta_q\ d^m Z\ d^m\Phi}{\hbar^{(n+m)}} \big((\kappa_G\cdot\kappa_H)\circ\iota_q^{-1}\big)(\xi_q\oplus Z)\\
        &\hspace{2in}\times \big((f\cdot g)\circ\iota_q^*\big)(\eta_q\oplus \Phi)e^{-i\big(\eta_q(\xi_q)+\Phi(Z)\big)/\hbar}
     \end{align}
for all $q\in M$ and $\xi_q\in T_qM$.  Now the fact that the Fourier transform $T_q^*M\to T_qM$ is one-to-one yields the result.\bigskip
\end{proof}

Next, define $\beta^M_0: C_c^\infty(\mathfrak{G}^*)\otimes C_c^\infty(\mathfrak{H}^*)\to C_c^\infty(T^*M)$ by
\begin{align}
    \beta_0(f\otimes g)(\eta_q) = (f\circ j_G)(\eta_q)\cdot (g\circ j_H)(\eta_q)
\end{align}
for all $f\in C_c^\infty(\mathfrak{G}^*)$, $g\in C_c^\infty(\mathfrak{H}^*)$, and $\eta_q\in T^*_qM$.

Recall further that it follows from the results of \citet{StFe21} that there are isomorphisms $\chi_G: C_0(\mathfrak{G}^*)\to \mathfrak{A}_0^G$ and $\chi_H: C_0(\mathfrak{H}^*)\to\mathfrak{A}_0^H$ given explicitly by the continuous extension of
\begin{align}
    \chi_G(f) = [\hbar\mapsto \mathcal{Q}_\hbar^G(f)] + K_0^G\\
    \chi_H(g) = [\hbar\mapsto \mathcal{Q}_\hbar^H(g)] + K_0^H
\end{align}
for $f\in C^\infty_{PW}(\mathfrak{G}^*)$ and $g\in C^{\infty}_{PW}(\mathfrak{H}^*)$.\bigskip

\begin{proposition}
\label{prop:kernel}
For all $f\in C_c^\infty(\mathfrak{G}^*)$ and $g\in C_c^\infty(\mathfrak{H}^*)$, we have 
\begin{align}
    \big(\chi_G(f)\otimes \chi_H(g)\big)\in \ker(\pi_0^G\otimes\pi_0^H)
    \end{align}
iff for almost all $q\in M$ and almost all $\eta_q\in T^*_qM$,
\begin{align}
    \beta_0^M(f\otimes g)(\eta_q) = 0.
\end{align}
\end{proposition}

\begin{proof}
    From Prop. \ref{prop:bimodule_limit}, we have that $(f\otimes g)\in \ker(\pi^G_0\otimes \pi^H_0)$ iff $\lim_{\hbar\to 0}\beta^{\xi_q}_\hbar(f\otimes g)(\eta_q) = 0$ for almost all $q\in M$, $\xi_q\in T_qM$, and $\eta_q\in T^*_qM$.  We have
    \begin{align}
        &\lim_{\hbar\to 0} \beta^{\xi_q}_\hbar(f\otimes g)(\eta_q)\\
        &=\lim_{\hbar\to 0}\int_{\mathcal{I}_q\oplus K^*_q}\frac{(E_q\circ \iota_q^{-1})(\xi_q\oplus Z)\ d^{m}Z\  d^{m}\Phi}{\hbar^m} \big((\kappa_G\cdot\kappa_H)\circ\iota_q^{-1}\big)(\xi_q\oplus Z)\\
    & \hspace{2in} \times \big((f\cdot g)\circ\iota_q^*\big)(\eta_q\oplus \Phi)e^{-i\Phi(Z)/\hbar}\\
    &= \lim_{\hbar\to 0}\int_{\mathcal{I}_q\oplus K^*_q}(E_q\circ \iota_q^{-1})(\xi_q\oplus Z)\ d^{m}Z\ d^{m}\Phi\  \big((\kappa_G\cdot\kappa_H)\circ\iota_q^{-1}\big)(\xi_q\oplus Z)\\
    & \hspace{2in} \times \big((f\cdot g)\circ\iota_q^*\big)(\eta_q\oplus \hbar\Phi)e^{-i\Phi(Z)}\\
    &=\big((f\cdot g)\circ\iota_q^*\big)(\eta_q\oplus 0)\\
    &\hspace{.5in} \times \int_{\mathcal{I}_q\oplus K^*_q}(E_q\circ \iota_q^{-1})(\xi_q\oplus Z)\ d^{m}Z\ d^{m}\Phi\  \big((\kappa_G\cdot\kappa_H)\circ\iota_q^{-1}\big)(\xi_q\oplus Z)
   e^{-i\Phi(Z)}.
    \end{align}
It follows that $(f\otimes g)\in\ker(\pi_0^G\otimes\pi_0^H)$ iff for almost all $q\in M$ and $\eta_q\in T_q^*M$,
\begin{align}
\label{eq:kernel_beta}
    0 = \big((f\cdot g)\circ\iota^*_q\big)(\eta_q\oplus 0) = \big((f\cdot g)\circ j_{|q}\big)(\eta_q) = (f\circ j_G)(\eta_q)\cdot (g\circ j_H)(\eta_q),
\end{align}
which is equivalent to $\beta_0^M(f\otimes g)(\eta_q) = 0.$
\end{proof}\bigskip

Since $\ker(\pi_0^G\otimes\pi_0^H)$ is a closed two-sided ideal, so is $\ker(\beta_0^M)$, and it follows that $\beta_0^M$ is a *-homomorphism.  This implies that $\beta_0^M$ continuously extends to a *-homomorphism on the completions $C_0(\mathfrak{G}^*)\otimes C_0(\mathfrak{H}^*)\to C_0(T^*M)$.  Hence, we have
\begin{align}
    C_0(S_M) \cong \big(C_0(\mathfrak{G}^*)\otimes C_0(\mathfrak{H}^*)\big)/\ker(\beta_0^M)\hookrightarrow C_0(T^*M),
\end{align}
which suffices to establish that there is a surjection $T^*M\to S_M$.

Now for the *-isomorphisms $\chi_G:C_0(\mathfrak{G}^*)\to \mathfrak{A}_0^G$ and $\chi_H: C_0(\mathfrak{H}^*)\to\mathfrak{A}^H_0$, denote the corresponding Gelfand duals by $\hat{\chi}_G: \mathcal{P}\big(\mathfrak{A}_0^G\big)\to \mathfrak{G}^*$ and $\hat{\chi}_H: \mathcal{P}\big(\mathfrak{A}_0^H\big)\to\mathfrak{H}^*$ for $\mathcal{P}\big(\mathfrak{A}_0^G\big)$ and $\mathcal{P}\big(\mathfrak{A}_0^H\big)$ the pure state spaces of $\mathfrak{A}_0^G$ and $\mathfrak{A}_0^H$, respectively. 
Since we have a diagram
\begin{align}
   \xymatrix{\mathcal{P}\big(\mathfrak{A}_0^G\big) & S_M\ar[l]\ar[r] & \mathcal{P}\big(\mathfrak{A}_0^H\big)},
\end{align}
we also have a diagram
\begin{align}
   \xymatrix{\mathfrak{G}^* & S_M\ar[l]\ar[r] & \mathfrak{H}^*}
\end{align}
so that we can think of $S_M$ as a subset of the product manifold $\mathfrak{G}^*\times\mathfrak{H}^*$.

Finally, notice that, per Eq.\eqref{eq:kernel_beta}, $\ker(\beta_0^M)$ corresponds to the set of functions that vanish after composition with $j_G$ or $j_H$ (really $j_{|q}$ for any $q\in M$).  From Gelfand duality, $S_M$ must therefore be isomorphic to the space on which every function in $\ker(\beta_0^M)$ vanishes, which is (after identifying $C_0(\mathfrak{G}^*)\otimes C_0(\mathfrak{H}^*)$ with $C_0(\mathfrak{G}^*\times\mathfrak{H}^*)$) the image of the pair map $(j_G,j_H)[T^*M]\subseteq \mathfrak{G}^*\times\mathfrak{H}^*$.  This suffices to establish that $S_M$ is isomorphic to the Lagrangian relation determined by the symplectic dual pair $T^*M$.\bigskip

\begin{proposition}
    $S_M\cong (j_G,j_H)[T^*M]$
\end{proposition}\bigskip

Recall that in \S\ref{sec:arrows}, we defined the quantization functor to map symplectic dual pairs in $\mathbf{LPoisson}$ to Hilbert bimodules in $\mathbf{LC^*}$ by the assignment
\begin{align}
    Q: [\mathfrak{G}^*\leftarrow T^*M\rightarrow \mathfrak{H}^*]\mapsto [C^*(G)\rightarrowtail \mathcal{E}_M\leftarrowtail C^*(H)].
\end{align}
\citet{La01a} showed this is indeed well-defined and functorial.  We furthermore defined the classical limit functor to map Hilbert bimodules in $\mathbf{LC^*}$ to symplectic dual pairs in $\mathbf{LPoisson}$ by the assignment
\begin{align}
    L: [C^*(G)\rightarrowtail \mathcal{E}_M\leftarrowtail C^*(H)]\mapsto [S_M],
\end{align}
where $[S_M]$ denotes the equivalence class of symplectic dual pairs determining the Lagrangian relation $S_M$.
Since $S_M\subseteq \mathfrak{G}^*\times\mathfrak{H}^*$ is the Lagrangian relation determined by the symplectic dual pair $T^*M$, the right-hand-side is indeed another expression for the equivalence class in $\mathbf{LPoisson}$ of symplectic dual pairs containing $T^*M$.
We have now shown that $L$ is well-defined: the quantization of any symplectic dual pair from the equivalence class $[S_M]$ must yield a Hilbert bimodule in the equivalence class $[\mathcal{E}_M]$ because those are precisely the bimodules whose classical limits return to $[S_M]$.  Furthermore the argument from \citet{FeSt24} already establishes that $L$ functorial, respective composition of Hilbert bimodules and Lagrangian relations by their corresponding tensor products.

Further, we have shown that $L\circ Q\left([T^*M]\right)\cong [T^*M]$.  It now follows, since in our definition of $\mathbf{LC^*}$ every arrow is obtained as the quantization of some arrow in $\mathbf{LPoisson}$, that we also have $Q\circ L\left([\mathcal{E}_M]\right)\cong [\mathcal{E}_M]$.  Thus, we have established the desired categorical equivalence.\bigskip

\begin{theorem}
    The functors $Q:\mathbf{LPoisson}\leftrightarrows\mathbf{LC^*}: L$ form a categorical equivalence.
\end{theorem}

\section{Conclusion}
\label{sec:con}

In this paper, we have established a one-to-one correspondence between certain equivalence classes of symplectic dual pairs and corresponding equivalence classes of Hilbert bimodules.  Symplectic dual pairs, and the Lagrangian relations they determine, can be understood as morphisms between classical phase spaces that preserve representation-theoretic structures.  Hilbert bimodules can be likewise understood as morphisms between algebras of quantum observables that preserve representation-theoretic structure.  So our one-to-one correspondence of morphisms encodes an equivalence of categories, as well as a type of equivalence of representation theories between models of classical and quantum physics.

Our technical results extend the procedure of \citet{FeSt24} to construct Lagrangian relations as the classical limit of a wider collection of (equivalence classes of) Hilbert bimodules.  This achievement obviates the need for extra conditions that ensure the resulting classical limit space is a symplectic manifold, reframing the classical limit in terms of Lagrangian relations rather than symplectic dual pairs.

Our results open further questions about the categories defined in this paper.  In $\mathbf{LPoisson}$, we motivated considering Lagrangian relations rather than symplectic dual pairs themselves because Lagrangian relations encode sufficient information to determine an induction functor from symplectic realizations of one Poisson manifold to those of another.  However, in $\mathbf{LC^*}$, we motivated considering tensor product kernels rather than Hilbert bimodules themselves because tensor product kernels are (merely) an invariant of the corresponding induction functor from Hilbert space representations of one C*-algebra to those of another.  How close is the analogy between Lagrangian relations and tensor product kernels?  Are there conditions that would uniquely determine an induction functor from a given tensor product kernel? Or are there other ways of defining morphisms in a category of C*-algebras altogether, so that morphisms fully determine induction functors between C*-algebras?  Finally, we mention a further open area of research suggested by \citet{La02a}, concerning whether one can generalize the quantization and classical limit functors that we have been working with to more general classes of morphisms inspired by $K$-theory for C*-algebras.

\newpage

\bibliography{bibliography.bib}

\end{document}